\documentclass[12pt]{article}

\usepackage[left=1in,right=1in,top=1in,bottom=1in]{geometry}
\usepackage{booktabs}

\usepackage{enumitem}
\setlist[description]{leftmargin=0em,labelindent=\parindent}

\usepackage{natbib}

\usepackage{amsmath,amsthm,amsfonts,amssymb}

\theoremstyle{definition}

\theoremstyle{definition}
\newtheorem{prop}{Proposition}

\usepackage{newtxtext,newtxmath}
\usepackage{mathtools}

\usepackage{graphicx}
\usepackage{caption}
\usepackage{subcaption}          
\usepackage[dvipsnames]{xcolor}
\usepackage{ulem}               
\usepackage{setspace}
\usepackage{soul}               
\usepackage[]{appendix}

\usepackage{titlesec}

\titleformat{\section}%
[block]% shape
{\normalsize\bfseries}%
{\thesection.}% label
{1em}% horizontal separation between label and title body
{\centering}% before the title body
[]% after the title body

\titlespacing*{\section}{0pt}{1em}{1em}

\titleformat{\subsection}%
[block]% shape
{\normalsize\bfseries}%
{\thesubsection}% {\thesection.\thesubsection}% label
{1em}% horizontal separation between label and title body
{\centering}% before the title body
[]% after the title body

\titlespacing*{\subsection}{0pt}{1em}{1em}

\newenvironment{figurenotes}[1][Note]{\begin{minipage}[t]{\linewidth}\footnotesize{\itshape#1: }}{\end{minipage}}

\definecolor{csub-blue}{RGB}{0, 53, 148}

\usepackage{url}
\usepackage[unicode,hyperfootnotes=false,linktocpage,pagebackref]{hyperref}
 \hypersetup{%
   colorlinks = true,
   pdfborder = {0 0 0},
   citecolor = csub-blue,
   linkcolor = csub-blue,
   urlcolor = csub-blue,
 }

\makeatletter
\def\@fnsymbol#1{\ensuremath{\ifcase#1\or \mathsection\or \mathparagraph\or \|\or **\or \dagger\dagger \or \ddagger\ddagger \else\@ctrerr\fi}}  
\makeatother

\title{\textsf{Responses of Unemployment to Productivity Changes\\for a General Matching Technology}}
\author{Rich Ryan\thanks{Email: \href{mailto:richryan@csub.edu}{RichRyan@csub.edu}.
  Department of Economics, California State University, Bakersfield, Bakersfield, CA 93311, USA.}}
\date{November 9, 2023}

\begin{document}
\maketitle
\thispagestyle{empty}

\begin{abstract}
Workers separate from jobs, search for jobs, accept jobs, and fund consumption with their wages.
Firms recruit workers to fill vacancies.
Search frictions prevent firms from instantly hiring available workers. 
Unemployment persists.
These features are described by the Diamond--Mortensen--Pissarides modeling framework.
In this class of models,
how unemployment responds to productivity changes depends on resources that can be allocated to job creation.
Yet, this characterization has been made when matching
is parameterized by a Cobb--Douglas technology.
For a canonical DMP model, I
(1)  demonstrate that a unique steady-state equilibrium will exist as long as the initial vacancy yields a positive surplus;
(2)  characterize responses of unemployment to productivity changes for a general matching technology; and
(3) show how a matching technology that is not Cobb--Douglas implies unemployment responds more to productivity changes,
which is independent of resources available for job creation,
a feature that will be of interest to business-cycle researchers.
\end{abstract}



\vfill

\pagebreak
\clearpage

\pagenumbering{arabic}

\section{Introduction}

Each month,
workers actively search for jobs and firms actively recruit workers.
Despite
workers being available to start jobs and
firms having posted job openings for vacant positions that could start within 30 days,
in any given month,
millions of unemployed workers cannot find jobs and millions of vacancies go unfilled.
It takes time for a worker to sift through job boards and fill out applications.
And
it is costly for a firm to post a vacancy.
Search frictions give rise to unemployment: a firm cannot instantly hire a worker.
A large class of models known as DMP models---short for Diamond--Mortensen--Pissarides models---use these features to study unemployment dynamics.

Within this class of models,
\citet{ljungqvist_sargent_2017} show that
the response of unemployment to changes in productivity
depends almost entirely on resources that can be allocated to vacancy creation.
To understand this idea, consider what happens when a job is created. 

Imagine that Eva is unemployed.
While unemployed,
Eva
searches for work,
cooks, cleans, and
collects unemployment-insurance benefits,
which in total amounts to $z$.
Upon finding a job, Eva uses the technology at the firm to produce $y$.
The match generates $y-z$.
This amount---at least some of it---can be allocated to vacancy creation.
\citet{ljungqvist_sargent_2017} call $y-z$ the fundamental surplus.
Viewed as a fraction of output, $\left( y-z \right) / y$, 
potential resources for vacancy creation are increasing in $y$.
The derivative with respect to $y$ is $z/y^{2}$.
The change will be large only when $z$ is large, which occurs when $z$ is close to $y$.
This is crucial:

The fundamental surplus must be small to produce high unemployment volatility
during business cycles, because
a change in productivity will generate a large change in
resources devoted to vacancy creation.
Other factors that affect unemployment 
can be ignored because they are bounded by
``a consensus about reasonable parameter values'' \citep[][2636]{ljungqvist_sargent_2017}.\footnote{Two essential contributions to the development of DMP models were made by \citet{pissarides_1985} and \citet{mortensen_pissarides_1994}.
  \citet{diamond_1982,diamond_1982b} made fundamental earlier contributions.
  See \citet{pissarides_2000}, the \citet{nobel_2010}, and \citet{petrosky-nadeau_wasmer_2017} for further background.}

There are two key issues:

\begin{description}
\item[Decomposition for general matching technology.]
  For the DMP class of models, \citet{ljungqvist_sargent_2017}
  establish that responses of unemployment to productivity changes
  depend on two factors.
  The two-factor decomposition, though, is based on workers and firms
  forming productive matches through a particular form of
  constant-returns-to-scale technology: Cobb--Douglas.
  Do conclusions about the
  two-factor, multiplicative decomposition hold for a general matching
  technology?
\item[Whether matching technology matters.]
  Under the decomposition,
  only one of the two factors is economically meaningful:
  the upper bound on resources, as a fraction of output,
  that the invisible hand can allocate to vacancy creation---the
  fundamental surplus fraction.
  The factor that does not matter includes an economy's matching technology.
  Is matching technology inconsequential for labor-market volatility?
\end{description}

The analysis here adopts the DMP framework and
a matching technology that exhibits constant returns to scale and satisfies standard regularity assumptions.
Within this framework,
a ratio of vacancies to unemployment---or labor-market tightness---drives unemployment dynamics.
I establish that \citeauthor{ljungqvist_sargent_2017}'s \citeyearpar{ljungqvist_sargent_2017}
two-factor, multiplicative decomposition of
the elasticity of tightness with respect to productivity holds
for a general matching technology.
The factor that includes the fundamental surplus matters, as predicted, but,
unlike the Cobb--Douglas case,
the second factor is not bounded by professional consensus.
Instead,
the second factor is bounded by
the elasticity of matching with respect to unemployment.
There is no reason for this bound to be constant over the business cycle and
the second factor---and therefore matching technology---could influence unemployment dynamics.

For conventional parameters, however,
I find that the fundamental surplus is significantly more meaningful.
Which suggests an economy's matching technology does not influence unemployment volatility.
Nevertheless, there is scope for matching technology to matter some.
Going beyond the local analysis of the elasticity at steady-state,
in a comparative steady-state analysis as a shortcut for analyzing model dynamics,
I show that a non-Cobb--Douglas matching technology
delivers larger responses of unemployment to productivity changes.
This matching technology, with nonconstant elasticity of matching with respect to unemployment,
offers a partial solution to the Shimer or unemployment-volatility puzzle,
the failure of the canonical DMP model to match the observed volatility of unemployment \citep{shimer_2005,pissarides_2009}.

In addition,
I provide an alternative interpretation of the existence and uniqueness
of an economy's non-stochastic, steady-state equilibrium.
An equilibrium will exist 
as long as it is profitable to post an initial vacancy.
As far as I know,
a unique equilibrium is typically posited to exist, which often holds for good economic reasons
(see, e.g., \citeauthor{pissarides_2000} [\citeyear[][19--20]{pissarides_2000}]). 
While my idea may be well understood by DMP researchers,
a benefit is providing an explicit requirement for parameter values and range in which equilibrium 
tightness is guaranteed to fall.

Understanding how a matching technology fits into
interpretations of the fundamental surplus is important.
The fundamental surplus, as \citet[][49]{ljungqvist_sargent_2021} emphasize,
offers a single channel for explaining how diverse features like
``sticky wages, elevated utility of leisure,
bargaining protocols that suppress the influence of outside values,
a frictional credit market that gives rise to a financial accelerator,
fixed matching costs, and government policies
like unemployment benefits and layoff costs''
can generate high unemployment volatility during business cycles.
All these features may interact with matching.

The DMP class of models explicates unemployment,
which is a source of misery for many people.
Understanding how matching affects unemployment dynamics in this class of models will help policymakers improve public policies that affect the unemployed.

\section{Model Environment}
\label{sol:model-environment}

To characterize unemployment dynamics,
I begin with a canonical DMP model.
The basic features include
linear utility,
random search,
workers with identical capacities for work,
wages determined as the outcome of Nash bargaining, 
job creation that drives the value of posting a vacancy to zero, and
exogenous separations.
The environment differs from the one studied by \citet{ljungqvist_sargent_2017} in a single way:
Instead of specifying the matching technology as Cobb--Douglas,
I use a general matching technology.
The
matching technology exhibits constant returns to scale in vacancies and the number of unemployed workers;
probabilities of finding and filling a job fall within $0$ and $1$; and
certain regularity conditions for limiting behavior hold.\footnote{After the ``preliminaries'' of describing the economic environment,
  \citet[][2635]{ljungqvist_sargent_2017}
  ``proceed under the assumption that the matching function has the Cobb--Douglas form, $M \left( u,v \right) = A u^{\alpha}v^{1-\alpha}$,
  where $A > 0$, and $\alpha \in \left( 0,1 \right)$ is the constant elasticity of matching with respect to unemployment,
  $\alpha = - q^{\prime} \left( \theta \right) \theta / q \left( \theta \right)$.''}

The main result establishes that
the fundamental insights of \citet{ljungqvist_sargent_2017} hold. 

\subsection{Description of the Model Environment}
\label{sec:model-description}

The environment is populated by a unit measure of identical, infinitely-lived workers.
Workers
are risk neutral with a discount factor of $\beta = \left(1+r\right)^{-1}$ and
are either employed or unemployed.
They aim to maximize discounted income.
Employed workers earn labor income.
Unemployed workers
earn no labor income,
look for work, and
experience the value of nonwork,
denoted by $z>0$.

The environment is also populated my a large measure of firms.
Firms are either active or inactive.
An active firm is either in a productive match with a worker or actively recruiting.
An inactive firm becomes active by posting a vacancy, which incurs a cost each period.

Once matched with a worker,
a firm operates a production technology that converts an indivisible unit of labor into $y$ units of output.
The production technology exhibits constant returns to scale in labor.
Each active firm matched with a worker employs a single worker.
While matched, a firm earns $y-w$, where $w$ is the per-period wage paid to the worker.
Wages are determined by the outcome of Nash bargaining.

All matches are exogenously destroyed with per-period probability $s$.
Free entry by the large measure of firms implies that a firm's expected discounted value of posting a vacancy equals zero.

A matching function $M$ determines the number of successful matches in a period.
Its arguments are the aggregate measures of unemployed workers, $u$, and vacancies, $v$.
The function $M \left( u,v \right)$ is increasing in both its arguments. 
More workers searching for jobs for a given level of vacancies leads to more matches and
more vacancies for a given level of unemployment leads to more matches.
In addition,
$M$ exhibits constant returns to scale in $u$ and $v$.

Labor-market tightness, $\theta$, is defined as the ratio of vacancies to unemployed workers, $\theta \coloneq v/u$.
Under random matching,
the probability that a firm fills a vacancy is given by
$q\left(\theta\right) \coloneq M\left(u,v\right)/v = M \left(\theta^{-1},1\right)$ and
the probability that an unemployed worker matches with a firm is given by
$\theta q\left(\theta\right) = M\left(u,v\right)/u = M\left(1,\theta\right)$.
Each unemployed worker faces the same likelihood of finding a job because
firms lack a recruiting technology that selects a particular candidate and workers do not direct their search effort.
The matching technology embodies frictions that generate involuntary unemployment.

\subsection{Key Bellman Equations}
\label{sec:key-bellm-equat}

Key Bellman equations in the economy include
a firm's value of a filled job and a posted vacancy; and
a worker's value of employment and unemployment.

A firm's value of a filled job, $\mathcal{J}$, and a posted vacancy, $\mathcal{V}$, satisfy
\begin{align}
\label{eq:J}
\mathcal{J} &= y-w + \beta\left[s\mathcal{V} + \left(1-s\right)\mathcal{J}\right] \\
\label{eq:V}  
\mathcal{V} &= -c+\beta\left\{ q\left(\theta\right) \mathcal{J} + \left[1-q\left(\theta\right)\right] \mathcal{V}\right\}.
\end{align}
The asset value of a filled job equals flow profit, $y-w$, plus the expected discounted value of continuing the match.
The match ends with probability $s$, providing the firm an opportunity to post a vacancy; and
the match endures with probability $1-s$, providing the value of a filled job.
The asset value of a vacancy equals the flow posting cost, $c$, plus the expected discounted value of matching
with a productive worker.
A productive match occurs with probability $q \left( \theta \right)$ and the vacancy remains unfilled the following period
with probability $1 - q \left( \theta \right)$.

A worker's value of employment, $\mathcal{E}$, and unemployment, $\mathcal{U}$, satisfy
\begin{align}
\label{eq:E}
\mathcal{E} &= w+\beta\left[sU+\left(1-s\right)\mathcal{E}\right] \\
  \label{eq:U}
\mathcal{U} &= z+\beta\left\{ \theta q\left(\theta\right)\mathcal{E}+\left[1-\theta q\left(\theta\right)\right]\mathcal{U}\right\}.
\end{align}
The asset value of employment equals the flow wage plus the expected discounted value of
being
unemployed with probability $1 - \theta q \left( \theta \right)$ or
employed with probability $\theta q \left( \theta \right)$ the following period.

Convention in the economy dictates that a worker and a firm split the surplus generated from a match through Nash bargaining.
Surplus from a match, $\mathcal{S}$, is
the benefit to a firm   from operating        as opposed to maintaining a vacancy plus
the benefit to a worker from earning a wage   as opposed to experiencing nonwork: $\mathcal{S}=\left(\mathcal{J}-\mathcal{V}\right)+\left(\mathcal{E}-\mathcal{U}\right)$.
Nash bargaining depends on the parameter $\phi \in \left[0,1\right)$, which measures a worker's relative bargaining power. 
The outcome of Nash bargaining specifies that
what the worker stands to gain equals their share of surplus and the firm receives the remainder:
$\mathcal{E}-\mathcal{U}=\phi\mathcal{S}$ and $\mathcal{J} - \mathcal{V} =\left(1-\phi\right)\mathcal{S}$.

The next section defines an equilibrium and establishes conditions for existence and uniqueness. 

\subsection{Equilibrium} 
\label{sec:equilibrium}

The size of the labor force is normalized to $1$.
A steady-state equilibrium requires that
the number of workers who separate from jobs, $s \left( 1-u \right)$, equals
the number of unemployed workers who find employment, $\theta q \left( \theta \right) u$,
so that the unemployment rate remains constant.
The steady-state condition implies $u = s / \left[ s + \theta q \left( \theta \right) \right]$, which yields
a Beveridge-curve relationship that is negative in $u$--$v$ space.
``When there are more vacancies, unemployment is lower because the unemployed find jobs more easily'' \citep[][20]{pissarides_2000}.
Consistent with this theory,
data on vacancies and unemployment exhibit a negative relationship.
The data are shown in figure \ref{fig:beveridge} of appendix \ref{sec:app:data-unempl-job}.\footnote{\citet{barlevy_etal_2023}
  provide a discussion of the negative relationship and \citet{elsby_michaels_ratner_2015}
  provide an overview within the context of the DMP framework.}

A steady-state equilibrium is a list of values $\left\langle u, \theta, w \right\rangle$ that 
satisfy the Bellman equations \eqref{eq:J}~--~\eqref{eq:U} along with
the free-entry condition that requires $\mathcal{V} = 0$,
the stipulation that wages are determined by the outcome of Nash bargaining, and
the steady-state unemployment rate.
These equations can be manipulated to yield a single expression in $\theta$ alone:
\begin{equation}
\label{eq:eqm-theta}
y-z = \frac{r+s+\phi\theta q\left(\theta\right)}{\left(1-\phi\right)q\left(\theta\right)}c. 
\end{equation}
Details for arriving at the expression in \eqref{eq:eqm-theta} are provided in appendix \ref{sec:app:deriv-fund-surpl}.

While \citet{ljungqvist_sargent_2017} do not explicitly establish the existence of a unique $\theta$ that solves \eqref{eq:eqm-theta},
it is straightforward to do so.
I state and sketch a proof here 
because the steps yield a requirement for parameters that has an intuitive interpretation.
Appendix \ref{sec:app:existence-uniqueness} provides more detail.

\begin{prop}
\label{prop:unique-theta}
  Suppose $y>z$, which says that workers produce more of the homogeneous consumption good at work than at home, and
  suppose that $\left(1-\phi\right)\left(y - z\right)/\left(r+s\right)>c$.
  Then a unique $\theta \in \left(0, \left( 1-\phi \right) \left( y-c \right) / \left( \phi c \right) \right)$ solves \eqref{eq:eqm-theta}. 
  The condition that $\left(1-\phi\right)\left(y-z\right)/\left(r+s\right)>c$ requires that the value of an initial job opening be positive.
\end{prop}

\begin{proof}
To establish existence, 
I define the function
\begin{align*}
  \mathcal{T}\left(\tilde{\theta}\right) = \frac{y-z}{c}
                                           - \frac{r+s+\phi\tilde{\theta}q\left(\tilde{\theta}\right)}{\left(1-\phi\right)q\left(\tilde{\theta}\right)}.
\end{align*}
Using the fact that $\lim_{\tilde{\theta} \rightarrow 0}q\left(\tilde{\theta}\right)=1$ and the requirement that the value of posting an initial vacancy is positive,
$\lim_{\tilde{\theta}\rightarrow0}\mathcal{T} \left(\tilde{\theta}\right) = \left( y-z \right) / c  - \left( r+s \right) / \left( 1-\phi \right) > 0$.
The inequality, as will be shown, can be interpreted as an initial posted vacancy having positive value.
Next, I define $\tilde{\theta}^{\bullet}= \left( 1-\phi \right) \left( y-z \right) / \left( \phi c \right) > 0$,
where the inequality comes from the fact that $y>z$ and $\phi \in \left[0, 1\right)$ by assumption.
Then $\mathcal{T}\left(\tilde{\theta}^{\bullet}\right) < 0$.
Because $\mathcal{T}$ is a combination of continuous functions,
it is also continuous.
An application of the intermediate-value theorem establishes that
there exists $\theta \in \left(0, \left( 1-\phi \right) \left( y-z \right) / \left( \phi c \right) \right)$ such that $\mathcal{T}\left(\theta\right) = 0$. 
Uniqueness follows from the fact that $\mathcal{T}$ is everywhere decreasing.

The condition that $\left(1-\phi\right)\left(y-z\right)/\left(r+s\right)>c$
requires that the value of posting an initial vacancy is profitable.
The following thought experiment illustrates why. 

Starting from a given level of unemployment,
which is guaranteed with exogenous separations,
the value of posting an initial vacancy is computed as $\lim_{\theta \rightarrow 0}\mathcal{V}$.
In the thought experiment,
the probability that the initial vacancy is filled is $1$, as $\lim_{\theta \rightarrow 0}q\left(\theta\right)=1$.
The following period the firm earns the value of a productive match,
which equals the flow payoff $y-w$ plus the value of a productive match discounted by $\beta\left(1-s\right)$.
The value of $\mathcal{J}$ is thus $\left( y-w \right) / \left[ 1-\beta\left(1-s\right) \right]$.

The wage rate paid by the firm in this scenario is $\lim_{\theta\rightarrow0}w = \phi y+\left(1-\phi\right)z$,
making $\mathcal{J} = \left(1-\phi\right)\left(y-z\right) / \left[ 1-\beta\left(1-s\right) \right]$.
Using this expression in the value of an initial vacancy yields
\begin{align*}
\lim_{\theta\rightarrow0}\mathcal{V} = \lim_{\theta\rightarrow0}\left\langle -c+\beta\left\{ q\left(\theta\right)\mathcal{J}+\left[1-q\left(\theta\right)\right]\mathcal{V}\right\} \right\rangle =-c+\beta\frac{\left(1-\phi\right)\left(y-z\right)}{1-\beta\left(1-s\right)}  > 0.
\end{align*}
The inequality stipulates that in order to start the process of posting vacancies,
the first vacancy needs to be profitable.
Developing this inequality yields $\left(1-\phi\right)\left(y-z\right) / \left( r+s \right) > c$,
establishing the condition listed in proposition \ref{prop:unique-theta}.

To arrive at the equilibrium, other profit-seeking firms post vacancies.
Filling a vacancy is no longer guaranteed, which raises the expected cost of maintaining a vacancy until it is filled, lowering $\mathcal{V}$.
Recruitment efforts eventually drive $\mathcal{V}$ to $0$.
\end{proof}

Checking that parameters generate an equilibrium can be useful,
especially in complicated models.
Not only
does a steady state permit a comparative-static analysis,
but not checking can lead to key channels being mistakenly downplayed \citep[][47n18]{christiano_eichenbaum_trabandt_2021,ljungqvist_sargent_2021}.
In addition,
knowing the interval that contains equilibrium tightness
is useful for finding its numerical value. 

\subsection{The Elasticity of Labor-Market Tightness with Respect to Productivity}
\label{sec:decomposition-fundamental-surplus}

Unemployment dynamics are driven by labor-market tightness within the DMP class of models.
Big responses of unemployment to the driving force of productivity
require a high elasticity of tightness with respect to productivity,
$\eta_{\theta,y}$.
My main result establishes that
the two-factor, multiplicative decomposition of $\eta_{\theta,y}$
holds for a general matching technology, not only for the Cobb--Douglas case.
The
result is stated in proposition \ref{prop:upsilon}.
Appendix \ref{sec:decomp-elast-mark} provides details.

\begin{prop}
  \label{prop:upsilon}
  In the canonical DMP search model,
  which features
  a general matching technology,
  random search, linear utility, workers with identical capacities for work,
  exogenous separations, and no disturbances in aggregate productivity,
  the elasticity of market tightness with respect to productivity can be decomposed as
  \begin{equation}
\label{eq:eta-theta-y}
  \eta_{\theta,y}= \left[1 + \frac{\left( r+s \right) \left( 1-\eta_{M,u} \right)}{\left( r+s \right)\eta_{M,u} + \phi \theta q \left( \theta \right)}  \right] \frac{y}{y-z}
  \eqcolon \Upsilon \frac{y}{y-z} < \frac{1}{\eta_{M,u}} \frac{y}{y-z},
\end{equation}
where the second factor is the inverse of fundamental surplus fraction and
the first factor is bounded below by $1$ and above by $1 / \eta_{M,u}$: $1 < \Upsilon < 1 / \eta_{M,u}$.
\end{prop}

The bound $1/\eta_{M,u}$ exceeds $1$ because
$\eta_{M,u} \in \left( 0,1 \right)$,
a well-known result that is proved
in appendix \ref{sec:app:elast-match-wrt-u} in proposition \ref{prop:eta-M-u} for completeness.

For the Cobb--Douglas case, $\eta_{M,u}$ is constant.
Estimates for its value and a
``consensus'' about reasonable values for the other terms in \eqref{eq:eta-theta-y}
imply that the factor $\Upsilon$ contributes little to the elasticity of market tightness
\citep[][2636]{ljungqvist_sargent_2017}.
Only the second factor, 
the inverse of the fundamental surplus fraction,
$y / \left( y-z \right)$,
can possibly generate unemployment dynamics observed in the data.
Because a diverse set of DMP models 
allow a similar two-factor decomposition,
the influence of the fundamental surplus
is a single, common channel for explaining unemployment volatility. 
Any additional feature added to a DMP model must run through this channel.

The decomposition, though, suggests that an economy's matching technology,
subsumed in $\Upsilon$,
does not matter for unemployment dynamics.
In general, however,
$\eta_{M,u}$ is not constant and depends on $\theta$, which varies meaningfully over the business cycle.

This variability is shown in figure \ref{fig:bound}, which depicts $1 / \eta_{M,u}$ for two prominent matching technologies.
While details will be provided in the computational experiment described below,
the main takeaway is that the bound warrants looking at whether matching technology can matter for unemployment volatility.
For the Cobb--Douglas parameterization, $1 / \eta_{M,u}$ equals the constant $1 / \alpha$.
In contrast,
the nonlinear series is depicted for values of $\theta$ observed in the US economy from December 2000 onwards.
The variability and magnitude of the series provide scope for investigating whether a matching technology matters for unemployment volatility.

\begin{figure}[h!]
\centerline{\includegraphics[width=0.8\textwidth]{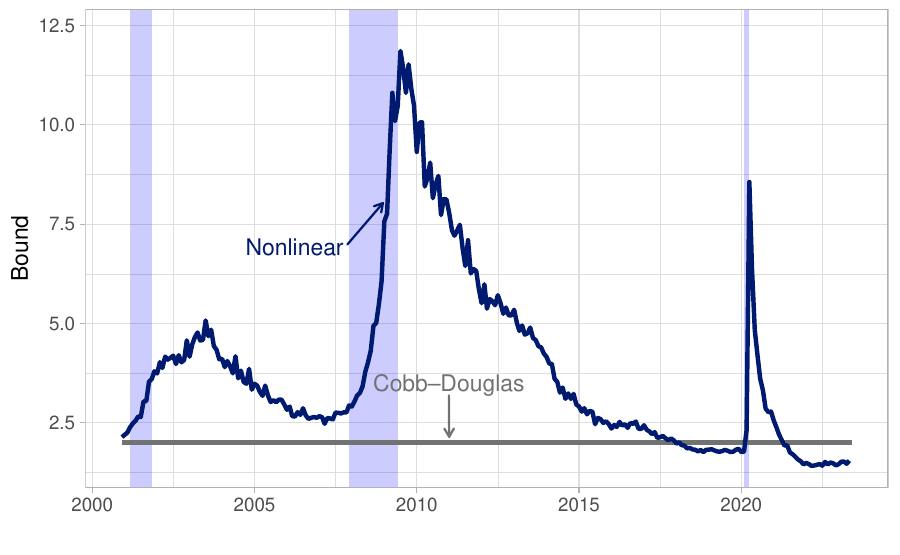}}
\caption[]{\label{fig:bound} Upper bounds for $\Upsilon$, inverses of the elasticity of matching with respect to unemployment, $1 / \eta_{M,u}$,
  for two matching technologies.}
\begin{figurenotes}[Notes]
  The inverses, $1 / \eta_{M,u}$, are upper bounds for $\Upsilon$, the first factor in the decomposition in \eqref{eq:eta-theta-y}.
  For the Cobb--Douglas technology, $1 / \eta_{M,u} = 1 / \alpha$, evaluated at $\alpha = 0.5$.
  For the nonlinear technology, $1 / \eta_{M,u} = \left( 1+\theta^{\gamma} \right) / \theta^{\gamma}$, evaluated at $\gamma = 1.27$ and
  values for $\theta$ observed in US data after December 2000.
  These parameter values are used in the computational experiment described in figure \ref{fig:du-dy}.
  Shaded areas indicate US recessions.
\end{figurenotes}
\begin{figurenotes}[Sources]
  Author's calculations using data from the US Bureau of Labor Statistics.
  Unemployment Level [UNEMPLOY], retrieved from FRED,
  Federal Reserve Bank of St.~Louis;
  \href{https://fred.stlouisfed.org/series/UNEMPLOY}{https://fred.stlouisfed.org/series/UNEMPLOY}.
  Job Openings: Total Nonfarm [JTSJOL], retrieved from FRED, Federal Reserve Bank of St.~Louis;
  \href{https://fred.stlouisfed.org/series/JTSJOL}{https://fred.stlouisfed.org/series/JTSJOL}.
\end{figurenotes}
\end{figure}

\section{Generating Larger Unemployment Responses to Productivity Perturbations}
\label{sec:gener-larg-unempl}

A comparative-equilibrium exercise is carried out by looking at how unemployment varies with productivity,
a main goal of DMP models, for two matching functions.
This shortcut for analyzing model dynamics is feasible because
unemployment is a fast-moving stock variable and productivity shocks exhibit high persistence.
The novelty is the comparison between matching functions.

I compare two parameterization of $M \left( u,v \right)$:
$A u^{\alpha} v^{1-\alpha}$ and $\mathcal{A} uv \left( u^{\gamma} + v^{\gamma} \right)^{-1/\gamma}$.
The first is the familiar  and empirically successful Cobb--Douglas parameterization
\citep{bleakley_fuhrer_1997,petrongolo_pissarides_2001}.
The second is a nonlinear parameterization suggested by \citet{den-haan_ramey_watson_2000}.
To motivate the nonlinear form,
imagine that each unemployed person contacts other agents randomly.
The probability that the other agent is a firm is $v / \left( u+v \right)$.
There are $u \times v / \left( u+v \right)$ matches.
The general form used here captures thick and thin market externalities.

The computational experiment begins by replicating \citet[][2644, fig.~2]{ljungqvist_sargent_2017}.
I use their parameters for comparison:
The model period is one day, avoiding job-finding and \mbox{-filling} probabilities above $1$.
The discount factor is $\beta = 0.95^{1/365}$, which corresponds to an annual interest rate of $5$ percent.
The daily separation rate is $s = 0.001$, which corresponds to a job lasting on average 2.8 years.
The bargaining parameter is $\phi = 0.5$, which is the midpoint of its range.
The flow cost of posting a vacancy is $c = 0.1$.
The value of nonwork is $z = 0.6$ and
values of workers' productivity are investigated above $z$ and substantially less than unity.\footnote{Appendix \ref{sec:joint-param-c-A}
  shows how a different choice of $c$ will produce the same equilibrium level of unemployment and job-finding (but not job-filling)
  through a different level of matching efficiency.
  \citet{kiarsi_2020} emphasizes the importance of the cost of posting a vacancy in this class of models.}

The remaining parameters specify the matching technology.
For the Cobb--Douglas matching technology 
$\alpha = 0.5$, so that $\eta_{M,u}$ equals the bargaining parameter.
This choice satisfies \citeauthor{hosios_1990}'s \citeyearpar{hosios_1990} efficiency condition.
So far,
these parameters agree with those adopted by \citet{ljungqvist_sargent_2017}.
For the nonlinear matching technology,
I set $\gamma = 1.27$ to agree with \citet[][491, table 1]{den-haan_ramey_watson_2000}.
The only remaining parameters to choose are the matching-efficiency parameters, $A$ and $\mathcal{A}$. 

The matching-efficiency parameters along
with productivity levels $y \in Y = \left\{ 0.61, 0.63, 0.65 \right\}$
index six economies. 
For each $y \in Y$,
$A$ and $\mathcal{A}$ are calibrated to make the unemployment rate 5 percent.
Some values of matching efficiency can cause finding and filling probabilities to rise above $1$, as established in appendix \ref{sec:app:two-matching-technologies},
but this is avoided by adopting a daily time period for the calibration \citep[][2639n6]{ljungqvist_sargent_2017}.
When I perturb productivity around $y$ for each economy, all other parameters remain fixed.

How the steady-state unemployment rate responds to productivity perturbations
is shown in  figure \ref{fig:du-dy}.
Unemployment
increases when productivity falls and
decreases when productivity rises,
regardless of 
matching technology.
Magnitudes of unemployment responses, however, depend on at least two features.

\begin{figure}[htbp]
\centerline{\includegraphics[width=\textwidth]{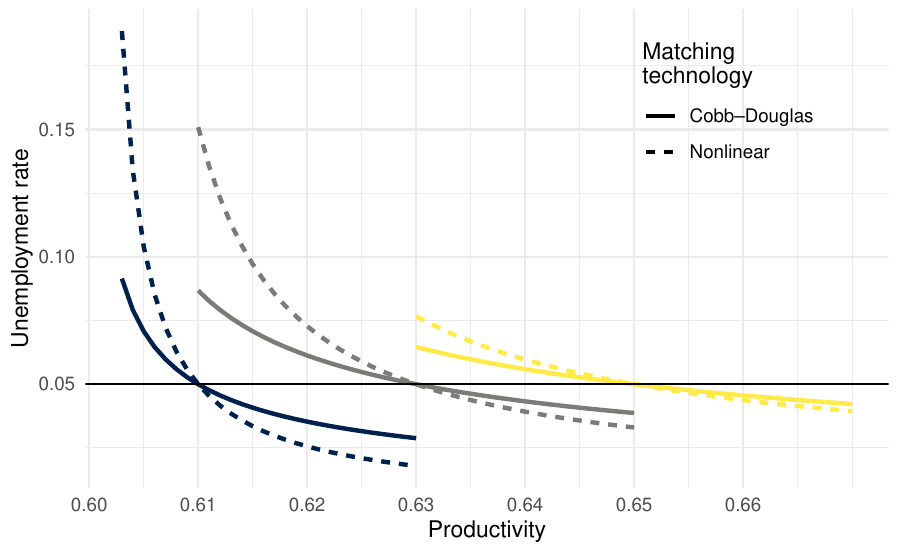}}
\caption[]{\label{fig:du-dy} Responses of unemployment to productivity perturbation.}
\begin{figurenotes}[Notes]
  The six economies are indexed by
  three productivity levels and
  two matching technologies.
  For each economy,
  matching efficiency for each matching technology is adjusted to generate
  5 percent unemployment at   
  productivity levels $0.61$, $0.63$, and $0.65$.
  The matching technologies are
  Cobb--Douglas, $M \left( u,v \right) = Au^{\alpha}v^{1-\alpha}$,
  and nonlinear, $\mathcal{M} \left( u,v \right) = \mathcal{A} uv \left( u^{\gamma} + v^{\gamma} \right)^{-1/\gamma}$.
  In each economy,
  steady-state unemployment rates are shown for
  perturbations in productivity around each economy's baseline productivity level.
\end{figurenotes}
\end{figure}

First,
looking from left to right,
the closer $y$ is to $z$, the smaller is the fundamental surplus.
The smaller the fundamental surplus, as predicted by equation \eqref{eq:eta-theta-y},
the more $\theta$ and thus unemployment respond to productivity changes.
For the two dark, navy curves at left, where line pattern indexes matching technology,
the fundamental surplus fraction is smallest and unemployment responses are largest.
For the two light, yellow curves at right, where line pattern again indexes matching technology,
the fundamental surplus fraction is largest and unemployment responses are smallest.

Second,
the relationship between unemployment and productivity depends on an economy's matching technology.
This point can be seen by comparing solid lines to broken lines.
For each $y \in Y$,
the nonlinear matching technology, 
causes unemployment to respond more to changes in productivity.
This result will be useful to those who
want address the Shimer or unemployment-volatility puzzle \citep{shimer_2005,pissarides_2009}.

\section{Conclusion}
\label{sol:conclusion}

For a canonical DMP model,
I showed that
the elasticity of labor-market tightness with respect to productivity can be decomposed into two multiplicative factors for a general matching technology.
One of the factors depends on the fundamental surplus and this factor has the largest influence on unemployment dynamics in the computational experiment.
The other factor is bounded above by the inverse of the elasticity of matching with respect to unemployment,
which, in general, varies meaningfully over the business cycle.
The finding leads to the conclusion that matching technology can matter for unemployment dynamics.
As features like sticky prices, sticky wages, idiosyncratic shocks, and composition effects of employment are added to the DMP class of models,
investigating interactions with the matching technology may be worth considering, not just the fundamental-surplus channel.

\bibliography{../../../bibliography/bibliography-org-ref}
\bibliographystyle{econ}

\clearpage
\pagebreak

\appendix
\appendixpage

\titleformat{\section}%
[block]% 
{\normalsize\bfseries}%
{Appendix \thesection.}% 
{1em}% 
{\centering}% 
[]% 

\titlespacing*{\section}{0pt}{1em}{1em}

\titleformat{\subsection}%
[block]% 
{\normalsize\bfseries}%
{Appendix \thesubsection}% 
{1em}% 
{\centering}% 
[]% 

\section{Data}

This section shares data.
Section \ref{sec:app:data-unempl-job} provides evidence for a few assertions made in the text,
including the simultaneous occurrence of millions of job opening and unemployed people.
Section \ref{sec:app:data-from-calibrated} shares data generated from the computational experiment.

Replication materials are available at
\begin{center}
  \href{https://github.com/richryan/fundamentalSurplusGeneralMatch}{https://github.com/richryan/fundamentalSurplusGeneralMatch}.
\end{center}

\subsection{Data on Unemployment and Job Openings}
\label{sec:app:data-unempl-job}

Each month
there are millions of unemployed people despite millions of job openings.
These two series are shown in figure \ref{fig:vac-unemp}.
The number of unemployed people
is a statistic computed from responses
to the Current Population Survey.
The series is depicted with the broken line.
The number of job openings
is a statistic computed from responses
to the Job Openings and Labor Turnover Survey.
The monthly series in figure \ref{fig:vac-unemp} start in December 2000,
when data from the Job Openings and Labor Turnover Survey become available.
The horizontal blue line indicates one million.

\begin{figure}[htbp]
\centerline{\includegraphics[width=0.8\textwidth]{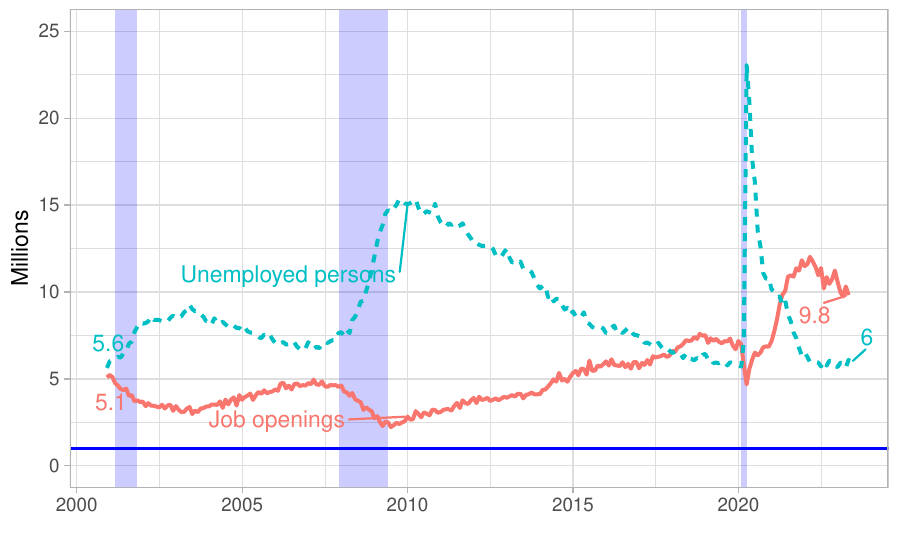}}
\caption[]{\label{fig:vac-unemp} Jobs openings and unemployed persons.}
\begin{figurenotes}[Notes]
  The blue horizontal line shows the level $1$ million.
  Shaded areas indicate US recessions.
\end{figurenotes}
\begin{figurenotes}[Sources]
  US Bureau of Labor Statistics.
  Unemployment Level [UNEMPLOY], retrieved from FRED,
  Federal Reserve Bank of St.~Louis;
  \href{https://fred.stlouisfed.org/series/UNEMPLOY}{https://fred.stlouisfed.org/series/UNEMPLOY}.
  Job Openings: Total Nonfarm [JTSJOL], retrieved from FRED, Federal Reserve Bank of St.~Louis;
  \href{https://fred.stlouisfed.org/series/JTSJOL}{https://fred.stlouisfed.org/series/JTSJOL}.
\end{figurenotes}
\end{figure}

When job openings or vacancies are plotted against unemployment
the relationship is known as the Beveridge curve.
Figure \ref{fig:beveridge} shows this relationship.
The relationship is negative because
more vacancies create more matches,
which reduce unemployment.
\citet{barlevy_etal_2023}
use a bathtub metaphor to describe the relationship
between vacancies and unemployment.
They also analyze longer time series,
which is informative.

\begin{figure}[htbp]
\centerline{\includegraphics[width=0.8\textwidth]{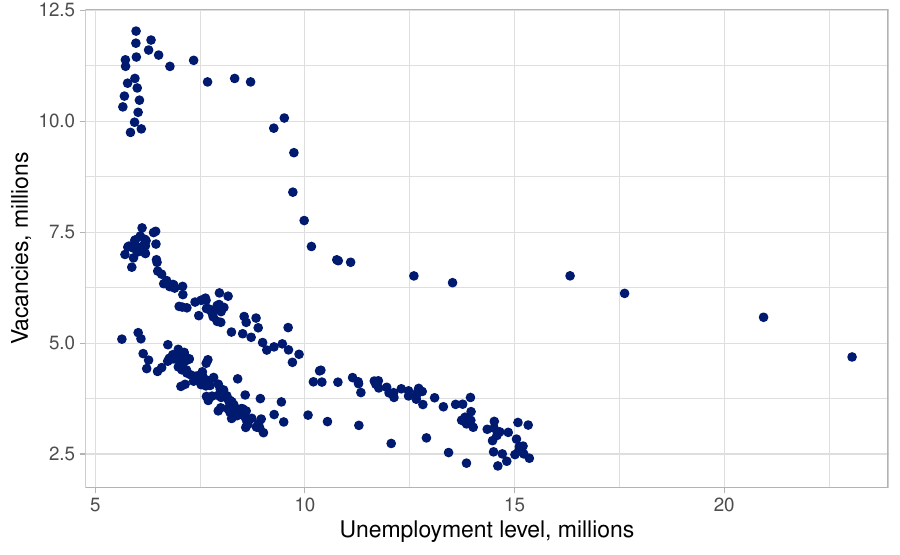}}
\caption[]{\label{fig:beveridge} Empirical relationship between vacancies and unemployment.}
\begin{figurenotes}[Note]
  The relationship is often referred to as the Beveridge curve.
  Vacancies refer to job openings.
\end{figurenotes}
\begin{figurenotes}[Sources]
  US Bureau of Labor Statistics.
  Unemployment Level [UNEMPLOY], retrieved from FRED,
  Federal Reserve Bank of St.~Louis;
  \href{https://fred.stlouisfed.org/series/UNEMPLOY}{https://fred.stlouisfed.org/series/UNEMPLOY}.
  Job Openings: Total Nonfarm [JTSJOL], retrieved from FRED, Federal Reserve Bank of St.~Louis;
  \href{https://fred.stlouisfed.org/series/JTSJOL}{https://fred.stlouisfed.org/series/JTSJOL}.
\end{figurenotes}
\end{figure}

The relationship in figure \ref{fig:beveridge} can be expressed in rates,
where
the unemployment rate is shown on the horizontal axis and
the vacancy rate is shown on the vertical axis.
The vacancy rate could be constructed by
dividing the number of vacancies by the sum of vacancies plus the aggregate measure of productive firms or employment.

\subsection{Data from Calibrated Models}
\label{sec:app:data-from-calibrated}

Figure \ref{fig:elasticities} provides data on the elasticities of market tightness and the wage rate for the six economies studied in figure \ref{fig:du-dy}.
In panel A, the elasticities of market tightness, computed using the expression in \eqref{eq:eta-theta-y},
show that the nonlinear technology delivers higher labor-market volatility.
This idea is expressed in figure \ref{fig:du-dy} in terms of unemployment rates.

Panel B of figure \ref{fig:elasticities} shares elasticities of the wage rate with respect to $y$ for the six economies studied in figure \ref{fig:du-dy}.
The values are computed using equation \eqref{eq:app:dw-dy} in this appendix.
The elasticities add an important point to any interpretation of higher labor-market volatility:
The economy indexed by $y = 0.61$ does not exhibit higher $\eta_{\theta,y}$ because wages respond less to productivity.
Under that false narrative,
the reason $\eta_{\theta,y}$ is higher would be because firms stand more to gain from an increase in productivity.
Because wages are less elastic and respond less to productivity, an increase in productivity would mean more profit for a firm owner.
The surplus generated from an increase in $y$ goes either to the worker or to the firm owner---and it does not go to the worker when wages are inelastic.
But figure \ref{fig:elasticities} rules this narrative out:
Panel B shows that
(1) wages are more elastic under the nonlinear technology than the Cobb--Douglas technology and
(2) wages are elastic across all productivity levels and matching technologies.
The data suggest that matching technology does matter.

Figures \ref{fig:job-finding} and \ref{fig:job-filling} show monthly job-finding and job-filling probabilities generated by the six economies.
The daily rates are converted to monthly rates using the computations discussed in appendix \ref{sec:app:converting-daily-to-monthly}.
The main takeaway is that the daily calibration forces monthly job-finding and filling rates to stay within $0$ and $1$; although,
the monthly job-filling rates are close to $1$.
\citeauthor{ljungqvist_sargent_2017}'s \citeyearpar{ljungqvist_sargent_2017} skillful suggestion is helpful, because
it is almost guaranteed that a high job-finding rate would push the job-filling rate above $1$ in a monthly or quarterly calibration.

\begin{figure}[htbp]
\centerline{\includegraphics[width=0.6\textwidth]{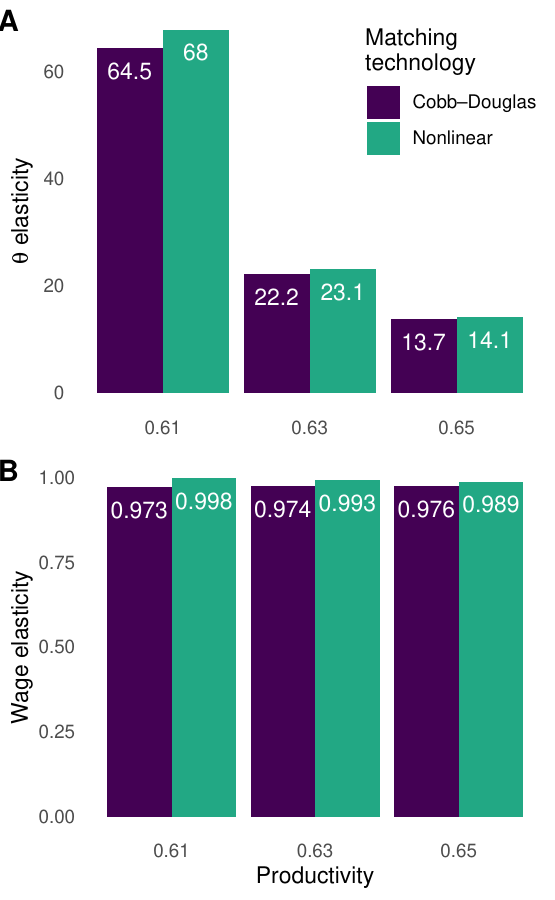}}
\caption[]{\label{fig:elasticities} Elasticities of market tightness and the wage rate.}
\begin{figurenotes}[Notes]
  The six values in each panel correspond to
  the six economies studied in figure \ref{fig:du-dy}.
  Panel A shows the elasticities of market tightness, $\eta_{\theta,y}$, computed using equation \eqref{eq:eta-theta-y}.
  Panel B shows the elasticities of the wage rate, $\eta_{w,y}$, computed using equation \eqref{eq:app:dw-dy}.
\end{figurenotes}
\end{figure}

\begin{figure}[htbp]
\centerline{\includegraphics[width=\textwidth]{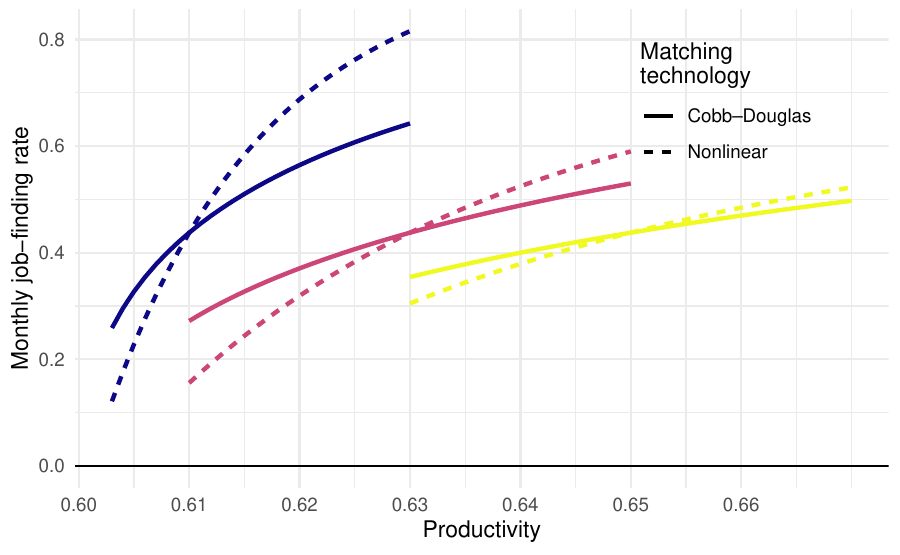}}
\caption[]{\label{fig:job-finding} Equilibrium rates of job finding.}
  \begin{figurenotes}[Notes]
    Job-finding rates for six economies with match efficiency adjusted to generate 5 percent unemployment for each matching technology at
    productivity levels $.61$, $0.63$, and $0.65$.
    The data correspond to the six economies studied in figure \ref{fig:du-dy}.
    Daily probabilities are converted to monthly probabilities using the computations described in appendix \ref{sec:app:converting-daily-to-monthly}.
  \end{figurenotes}
\end{figure}

\begin{figure}[htbp]
  \centerline{\includegraphics[width=\textwidth]{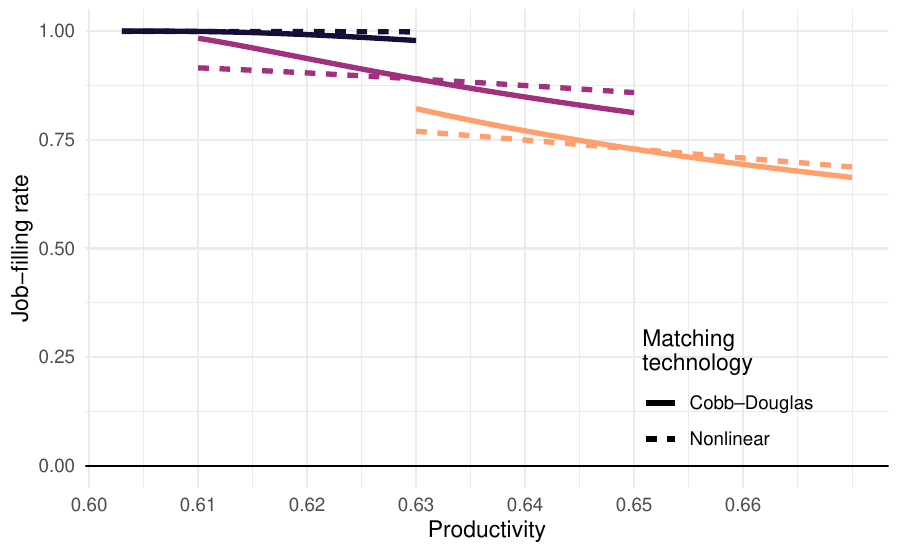}}
\caption[]{\label{fig:job-filling} Equilibrium rates of job filling.}
\begin{figurenotes}[Notes]
    Job-filling rates for six economies with match efficiency adjusted to generate 5 percent unemployment for each matching technology at
    productivity levels $.61$, $0.63$, and $0.65$.
    The data correspond to the six economies studied in figure \ref{fig:du-dy}.
    Daily probabilities are converted to monthly probabilities using the computations described in appendix \ref{sec:app:converting-daily-to-monthly}.    
  \end{figurenotes}
\end{figure}

\clearpage
\pagebreak

\section{The Elasticity of Matching with Respect to Unemployment}
\label{sec:app:elast-match-wrt-u}

In general, a matching technology computes the number of new matches
or new hires produced when $u$ workers are searching for jobs and
$v$ vacancies are posted.
A matching technology, $M$, in other words,
maps unemployment and vacancies into matches:
$M: \mathbb{R}_{+} \times \mathbb{R}_{+} \rightarrow \mathbb{R}_{+}$.
It is increasing in
both its non-negative arguments and exhibits constant returns to scale in $u$ and $v$.

Before turning to the particular parameterization,
for completeness,
I state and prove a well known result about the elasticity of matching with respect
to unemployment.
This result will be used in the discussion about
the decomposition of the elasticity of tightness.

I use the following notation: 
\begin{itemize}
\item $M$ denotes the number of new matches generated within a period.
\item $u$ denotes the number of unemployed workers searching for a job. 
\item $v$ denotes the number of vacancies posted by firms recruiting workers.
\item $\theta=v/u$, the ratio of vacancies to unemployment, denotes labor-market
tightness. 
\item $q\left(\theta\right)=M/v$ denotes the probability that a vacancy
is filled. 
\item $\theta q\left(\theta\right)=M/u$ denotes the probability that a
worker finds a job. 
\end{itemize}
That $M/v$ denotes the probability that a posted vacancy is filled
follows from the assumption that search is random, meaning each vacancy
faces the same likelihood of being filled.

In addition, a general matching technology should possess the following characteristics:
\begin{equation}
\lim_{\theta\rightarrow0}q\left(\theta\right)=1\text{ and }\lim_{\theta\rightarrow\infty}q\left(\theta\right)=0\label{eq:app:limit-fill}
\end{equation}
while 
\begin{equation}
\lim_{\theta\rightarrow0}\theta q\left(\theta\right)=0\text{ and }\lim_{\theta\rightarrow\infty}\theta q\left(\theta\right)=1,\label{eq:app:limit-find}
\end{equation}
which says that the job-filling probability goes to $1$ as the ratio of job openings to unemployed persons goes to $0$,
or $v/u \rightarrow 0$.
Likewise, it is nearly impossible to fill a vacancy when there are many job openings relative to the number of unemployed.
Job-finding is the flip side of this process,
which explains the limits in \eqref{eq:app:limit-find}.

The elasticity of matching with respect to unemployment is
the percent change in matches given a percent change in unemployment:
\begin{equation}
\eta_{M,u} \coloneq \frac{dM}{du}\frac{u}{M}= -\frac{\theta q^{\prime}\left(\theta\right)}{q\left(\theta\right)}.\label{eq:app:eta-M-u}
\end{equation}
The expression in \eqref{eq:app:eta-M-u} comes from direct computation.
Indeed, from the definition of job filling, $q\left(\theta\right)=M/v$, it follows that
\begin{align*}
\frac{dM}{du} &= \frac{d}{du}\left[q\left(\theta\right)v\right]\\
 & =q^{\prime}\left(\theta\right)\frac{-v}{u^{2}}v=-q^{\prime}\left(\theta\right)\theta^{2},
\end{align*}
where the first equality uses the fact that $q\left(\theta\right)=M/v$
and the second line uses the fact that $d\theta/du=-v/u^{2}$.
Thus
\begin{align*}
\frac{dM}{du}\frac{u}{M} & =-\frac{q^{\prime}\left(\theta\right)\theta^{2}}{M/u}=-\frac{q^{\prime}\left(\theta\right)\theta^{2}}{\theta q\left(\theta\right)}\\
 & =-\frac{\theta q^{\prime}\left(\theta\right)}{q\left(\theta\right)}>0,
\end{align*}
where the last equality in the first line uses the definition of job finding: $\theta q\left(\theta\right)=M/u$.
The inequality uses the property that $q^{\prime}<0$.
The inequality $\eta_{M,u}>0$ means that, for a given level of labor demand,
an increase in workers searching for jobs increases the number of new hires.

Moreover $\eta_{M,u}$ lies in the interval $\left(0,1\right)$.
It has already been established that $\eta_{M,u}>0$.
The fact that $\eta_{M,u}<1$ can be established by differentiation of $\theta q\left(\theta\right)$ with respect to $v$: 
\begin{align*}
\frac{d}{dv}\left[\theta q\left(\theta\right)\right] & =\left\{ \left[1\times q\left(\theta\right)\right]+\theta q^{\prime}\left(\theta\right)\right\} \frac{1}{u}\\
 & =\left[q\left(\theta\right)+\theta q^{\prime}\left(\theta\right)\right]\frac{1}{u}.
\end{align*}
Because $\theta q\left(\theta\right)$ can be written $\theta M\left(u,v\right)/v$, it also true that
\begin{align*}
\frac{d}{dv}\left[\theta q\left(\theta\right)\right] & =\frac{1}{u}\frac{M}{v}+\theta\frac{M_{v}v-M}{v^{2}}\\
 & =\frac{1}{u}\frac{M}{v}+\theta\left(\frac{M_{v}}{v}-\frac{q\left(\theta\right)}{v}\right)\\
 & =\frac{1}{v}\left\{ \theta q\left(\theta\right)+\theta\left[M_{v}-q\left(\theta\right)\right]\right\} 
\end{align*}
where $M_{v}$ is the derivative of the matching function with respect to vacancies.
Combining these two expressions yields 
\begin{align*}
\left[q\left(\theta\right)+\theta q^{\prime}\left(\theta\right)\right]\frac{1}{u} & =\frac{1}{v}\left\{ \theta q\left(\theta\right)+\theta\left[M_{v}-q\left(\theta\right)\right]\right\} \\
\therefore\left[q\left(\theta\right)+\theta q^{\prime}\left(\theta\right)\right]\theta & =\theta q\left(\theta\right)+\theta\left[M_{v}-q\left(\theta\right)\right]\\
\therefore\theta q^{\prime}\left(\theta\right) & =M_{v}-q\left(\theta\right)\\
\therefore\frac{\theta q^{\prime}\left(\theta\right)}{q\left(\theta\right)} & =\frac{M_{v}}{q\left(\theta\right)}-1\\
\therefore1-\left(-\frac{\theta q^{\prime}\left(\theta\right)}{q\left(\theta\right)}\right) & =\frac{M_{v}}{q\left(\theta\right)}\\
\therefore1-\eta_{M,u} & =\frac{M_{v}}{q\left(\theta\right)}.
\end{align*}
Therefore, $1-\eta_{M,u}$ is positive since $M$ is increasing in both its arguments.
Re-arranging $1-\eta_{M,u} > 0$ establishes that $\eta_{M,u}<1$.
Hence, $\eta_{M,u}\in\left(0,1\right)$.
These results are collected in proposition \ref{prop:eta-M-u}.

\begin{prop}
  \label{prop:eta-M-u}
  Given a constant-returns to scale matching technology
  that is increasing in both its arguments, $M\left(u,v\right)$,
  the elasticity of matching with respect to unemployment,
  $\eta_{M,u} \coloneq -\theta q^{\prime}\left(\theta\right)/q\left(\theta\right)$,
  lies in the interval $\left(0,1\right)$.
\end{prop}

\clearpage
\pagebreak

\section{Properties of Two Prominent Matching Technologies}
\label{sec:app:two-matching-technologies}

\textbf{Cobb--Douglas}.
One prominent parameterization of $M$ is
\begin{equation}
\mathsf{M}\left(u,v\right)=\mathsf{A}u^{\alpha}v^{1-\alpha},\quad \mathsf{A}>0,\quad\alpha\in\left(0,1\right).\label{eq:app:M-cobb-douglas}
\end{equation}
The function $\mathsf{M}$ exhibits constant returns to scale in $u$
and $v$. This is the familiar Cobb--Douglas parameterization. Because
$\mathsf{M}$ is increasing in both unemployment and vacancies, $\alpha>0$
and $1-\alpha>0$. These two inequalities imply $0<\alpha<1$.
The Cobb--Douglas parameterization delivers good empirical performance based on the statistical evidence provided by \citet{petrongolo_pissarides_2001}.

Under random search, the probability that a vacancy is filled is $\mathsf{M}/v$:
\begin{align*}
q_{\mathsf{A}}\left(\theta\right) & =\frac{\mathsf{M}}{v}=\mathsf{A} u^{\alpha}v^{-\alpha}=\mathsf{A}\theta^{-\alpha},
\end{align*}
where the notation $q_{\mathsf{A}}$ explicitly references the matching efficiency parameter, $\mathsf{A}$.
A direct computation establishes that
the job-filling probability is decreasing in tightness.
It is harder, in other words, for a firm to fill a vacancy the more vacancies there are for a given level of unemployment.

Under random search, the probability that a worker finds a job is
\begin{align*}
\theta q_{\mathsf{A}}\left(\theta\right)=\frac{\mathsf{M}}{u}\frac{v}{v}=\mathsf{A}\theta^{1-\alpha}.
\end{align*}
The job-finding probability is increasing in $\theta$.
It is easier, in other words, for an individual worker to find a job the more vacancies there are for a given level of unemployment.

The elasticity of matching with respect to unemployment is constant:
\begin{align}
\label{eq:app:eta-M-u-Cobb-Douglas}
  \begin{split}
\eta_{M,u} & =-\frac{\theta q^{\prime}_{\mathsf{A}}\left(\theta\right)}{q_{\mathsf{A}}\left(\theta\right)}\\
 & =-\frac{\theta\left(-\alpha\right)\mathsf{A}\theta^{-\alpha-1}}{\mathsf{A}\theta^{-\alpha}}\\
 & =\alpha.    
  \end{split}
\end{align}

\textbf{Nonlinear}.
Another parameterization of the matching technology, suggested by \citet{den-haan_ramey_watson_2000}, is
\begin{equation}
\label{eq:app:M-nonlinear}
\mathcal{M}\left(u,v\right)=\mathcal{A}\frac{uv}{\left[u^{\gamma}+v^{\gamma}\right]^{1/\gamma}},\quad\mathcal{A},\gamma>0.
\end{equation}
The function $\mathcal{M}$ exhibits constant returns to scale:
For any $\lambda\in\mathbb{R}_{+}$,
\begin{align*}
\mathcal{A}\frac{\left(\lambda u\right)\left(\lambda v\right)}{\left[\left(\lambda u\right)^{\gamma}+\left(\lambda v\right)^{\gamma}\right]^{1/\gamma}} & =\mathcal{A}\frac{\lambda^{2}uv}{\left\{ \lambda^{\gamma}\left[u^{\gamma}+v^{\gamma}\right]\right\} ^{1/\gamma}}\\
 & =\mathcal{A}\frac{\lambda^{2}uv}{\lambda\left[u^{\gamma}+v^{\gamma}\right]^{1/\gamma}}\\
 & =\lambda\mathcal{A}\frac{uv}{\left[u^{\gamma}+v^{\gamma}\right]^{1/\gamma}}.
\end{align*}
In addition, $\mathcal{M}$ is increasing in both its arguments. Indeed,
\begin{align*}
\frac{\partial\mathcal{M}}{\partial u} & =\mathcal{A}\frac{v\left(u^{\gamma}+v^{\gamma}\right)^{1/\gamma}-\frac{1}{\gamma}\left(u^{\gamma}+v^{\gamma}\right)^{1/\gamma-1}\gamma u^{\gamma-1}uv}{\left(u^{\gamma}+v^{\gamma}\right)^{\frac{2}{\gamma}}}\\
 & =\mathcal{A}\frac{v\left[u^{\gamma}+v^{\gamma}\right]^{1/\gamma}-\left(u^{\gamma}+v^{\gamma}\right)^{1/\gamma-1}u^{\gamma}v}{\left(u^{\gamma}+v^{\gamma}\right)^{\frac{2}{\gamma}}}\\
 & =\mathcal{A}\frac{u\left[u^{\gamma}+v^{\gamma}\right]^{1/\gamma}}{\left(u^{\gamma}+v^{\gamma}\right)^{\frac{2}{\gamma}}}\left(1-\frac{u^{\gamma}}{u^{\gamma}+v^{\gamma}}\right)\\
 & >0.
\end{align*}
A symmetric argument establishes that $\mathcal{M}$ is increasing in $v$. 

Under the nonlinear parameterization, the probability that a vacancy is filled is
\begin{align*}
  q_{\mathcal{A}}\left(\theta\right)
  &= \frac{\mathcal{M}}{v}=\mathcal{A}\frac{u}{\left[u^{\gamma}+v^{\gamma}\right]^{1/\gamma}}\frac{1/u}{1/u} \\
  & =\mathcal{A}\frac{1}{\left[1+\left(v/u\right)^{\gamma}\right]^{1/\gamma}}\\
 & =\mathcal{A}\frac{1}{\left(1+\theta^{\gamma}\right)^{1/\gamma}}.
\end{align*}
A direct computation establishes that the job-filling probability is decreasing in tightness:
\begin{align*}
\frac{dq_{\mathcal{A}}\left(\theta\right)}{d\theta}=-\frac{\mathcal{A}}{\gamma}\frac{1}{\left(1+\theta^{\gamma}\right)^{1/\gamma-1}}\gamma\theta^{\gamma-1}<0.
\end{align*}
The probability a worker finds a job is
\begin{align*}
\theta q_{\mathcal{A}}\left(\theta\right) & =\frac{\mathcal{M}}{u}=\mathcal{A}\frac{\theta}{\left(1+\theta^{\gamma}\right)^{1/\gamma}} > 0.
\end{align*}
The job-finding probability under the nonlinear parameterization is increasing in $\theta$:
\begin{align*}
\frac{d}{d\theta}\left[\theta q_{\mathcal{A}}\left(\theta\right)\right] & =\mathcal{A}\frac{\left(1+\theta^{\gamma}\right)^{1/\gamma}-\theta\frac{1}{\gamma}\left(1+\theta^{\gamma}\right)^{1/\gamma-1}\gamma\theta^{\gamma-1}}{\left(1+\theta^{\gamma}\right)^{2/\gamma}}\\
 & =\mathcal{A}\frac{\left(1+\theta^{\gamma}\right)^{1/\gamma}-\left(1+\theta^{\gamma}\right)^{1/\gamma-1}\theta^{\gamma}}{\left(1+\theta^{\gamma}\right)^{2/\gamma}}\\
 & =\mathcal{A}\frac{\left(1+\theta^{\gamma}\right)^{1/\gamma}}{\left(1+\theta^{\gamma}\right)^{2/\gamma}}\left(1-\frac{\theta^{\gamma}}{1+\theta^{\gamma}}\right)\\
 & >0.
\end{align*}

For the nonlinear matching technology, when $\mathcal{A<\infty}$,
the job-finding probability is between $0$ and $\mathcal{A}$.
Indeed,
\begin{align*}
\lim_{\theta\rightarrow0}\theta q_{\mathcal{A}}\left(\theta\right)=\lim_{\theta\rightarrow0}\mathcal{A}\frac{\theta}{\left(1+\theta^{\gamma}\right)^{1/\gamma}}=0
\end{align*}
and 
\begin{align*}
\lim_{\theta\rightarrow\infty}\theta q_{A}\left(\theta\right) & =\lim_{\Theta\rightarrow\infty}\mathcal{A}\frac{\theta}{\left(1+\theta^{\gamma}\right)^{1/\gamma}}=\lim_{\theta\rightarrow\infty}\mathcal{A}\frac{1}{\left(1+\theta^{\gamma}\right)^{1/\gamma-1}\theta^{\gamma-1}}\\
 & =\mathcal{A},
\end{align*}
where the second-to-last equality uses L'H\^{o}pital's rule and the
fact that
\begin{align*}
\left(1+\theta^{\gamma}\right)^{1/\gamma-1}\theta^{\gamma-1} & =\left(1+\theta^{\gamma}\right)^{\frac{1-\gamma}{\gamma}}\left(\frac{1}{\theta}\right)^{1-\gamma}\\
 & =\left(1+\theta^{\gamma}\right)^{\frac{1-\gamma}{\gamma}}\left(\frac{1}{\theta}\right)^{1-\gamma}\\
 & =\left(1+\theta^{\gamma}\right)^{\frac{1-\gamma}{\gamma}}\left[\left(\frac{1}{\theta}\right)^{\gamma}\right]^{\frac{1-\gamma}{\gamma}}\\
 & =\left(1+\theta^{\gamma}\right)^{\frac{1-\gamma}{\gamma}}\left(\frac{1}{\theta^{\gamma}}\right)^{\frac{1-\gamma}{\gamma}}\\
 & =\left[\frac{1}{\theta^{\gamma}}\left(1+\theta^{\gamma}\right)\right]^{\frac{1-\gamma}{\gamma}}=\left[1+\frac{1}{\theta^{\gamma}}\right]^{\frac{1-\gamma}{\gamma}}
\end{align*}
and therefore
\begin{align*}
\lim_{\theta\rightarrow\infty}\left(1+\theta^{\gamma}\right)^{1/\gamma-1}\theta^{\gamma-1}=\lim_{\theta\rightarrow\infty}\left[1+\frac{1}{\theta^{\gamma}}\right]^{\frac{1-\gamma}{\gamma}}=1.
\end{align*}
In addition, the fact that the job-finding probability is increasing everywhere implies that the probability a worker finds a job lies between $0$ and $1$ when $\mathcal{A}=1$. 

Similarly,
the job-filling probability for the nonlinear parameterization falls between $0$ and $\mathcal{A}$.
Indeed,
\begin{align*}
\lim_{\theta\rightarrow\infty}q_{A}\left(\theta\right)=\lim_{\theta\rightarrow\infty}\mathcal{A}\frac{1}{\left(1+\theta^{\gamma}\right)^{1/\gamma}}=0
\end{align*}
and
\begin{align*}
\lim_{\theta\rightarrow0}q_{A}\left(\theta\right)=\lim_{\theta\rightarrow\infty}\mathcal{A}\frac{1}{\left(1+\theta^{\gamma}\right)^{1/\gamma}}=\mathcal{A}.
\end{align*}
The fact that the job-filling probability is decreasing everywhere
implies that the probability a job is filled falls between $0$ and $\mathcal{A}$.

The elasticity of matching with respect to unemployment for the nonlinear
parameterization is
\begin{align}
  \label{eq:app:eta-M-u-nonlinear}
  \begin{split}
\eta_{M,u} & =-\frac{\theta q^{\prime}\left(\theta\right)}{q\left(\theta\right)}\\
 & =\frac{\theta\frac{1}{\gamma}\mathcal{A}\left(1+\theta^{\gamma}\right)^{-1/\gamma-1}\gamma\theta^{\gamma-1}}{\mathcal{A}\left(1+\theta^{\gamma}\right)^{-1/\gamma}}\\
 & =\frac{\left(1+\theta^{\gamma}\right)^{-1/\gamma-1}\theta^{\gamma}}{\left(1+\theta^{\gamma}\right)^{-1/\gamma}}\\
 &= \left(1+\theta^{\gamma}\right)^{-1}\theta^{\gamma} \\
  &= \frac{\theta^{\gamma}}{1+\theta^{\gamma}}.    
  \end{split}
\end{align}
As implied by proposition \ref{prop:eta-M-u}, the elasticity in \eqref{eq:app:eta-M-u-nonlinear} falls
inside the unit interval.
Unlike the Cobb-Douglas parameterization, $\eta_{M,u}$ is not constant.

\textbf{Minor discussion}.
While Cobb--Douglas fits the data well, as \citet{petrongolo_pissarides_2001}  and \citet{bleakley_fuhrer_1997} have shown,
not all specifications keep job-finding and job-filling probabilities within the unit interval \citep{den-haan_ramey_watson_2000}.
This feature is one motivation for using the nonlinear technology in business-cycle research like that in \citet{petrosky-nadeau_zhang_2017}.
Although,
\citet[][2639n6]{ljungqvist_sargent_2017} skillfully show how a daily calibration could avoid this outcome and encourage firms to post vacancies.

\clearpage
\pagebreak

\section{Derivations for the Fundamental Surplus Omitted from the Main Text}
\label{sec:app:deriv-fund-surpl}

In this section,
I derive expressions presented in sections \ref{sec:equilibrium} and \ref{sec:decomposition-fundamental-surplus}.
And I provide further details used in the proofs of propositions \ref{prop:unique-theta} and \ref{prop:upsilon}.
Many of the expressions are repeated here so that I can explicitly refer to them.

\subsection{Key Bellman Equations}

Here I repeat the key Bellman equations for the canonical DMP model.

Key Bellman equations in the economy for firms are 
\begin{equation}
\mathcal{J}=y-w+\beta\left[s\mathcal{V}+\left(1-s\right)\mathcal{J}\right],\label{eq:app:J}
\end{equation}
\begin{equation}
\mathcal{V}=-c+\beta\left\{ q\left(\theta\right)\mathcal{J}+\left[1-q\left(\theta\right)\right]\mathcal{V}\right\} .\label{eq:app:V}
\end{equation}
Imposing the zero-profit condition in equation \eqref{eq:app:V} implies
\begin{align*}
0 & =-c+\beta\left\{ q\left(\theta\right)\mathcal{J}+\left[1-q\left(\theta\right)0\right]\right\} \\
\therefore c & =\beta q\left(\theta\right)\mathcal{J}
\end{align*}
or 
\begin{equation}
\mathcal{J}=\frac{c}{\beta q\left(\theta\right)}.\label{eq:app:J-eqm}
\end{equation}
Substituting this result into equation \eqref{eq:app:J} and imposing
the zero-profit condition implies 
\begin{align*}
\mathcal{J} & =y-w+\beta\left[s\mathcal{V}+\left(1-s\right)\mathcal{J}\right]\\
\therefore\frac{c}{\beta q\left(\theta\right)} & =y-w+\beta\left(1-s\right)\frac{c}{\beta q\left(\theta\right)}\\
\therefore\frac{c}{\beta q\left(\theta\right)} & =y-w+\left(1-s\right)\frac{c}{q\left(\theta\right)}\\
\therefore w & =y+\left(1-s\right)\frac{c}{q\left(\theta\right)}-\frac{c}{\beta q\left(\theta\right)}\\
\therefore w & =y+\frac{c}{q\left(\theta\right)}\left(1-s-\frac{1}{\beta}\right)\\
\therefore w & =y+\frac{c}{q\left(\theta\right)}\left(-s-r\right).
\end{align*}
This simplifies to 
\begin{equation}
w=y-\frac{r+s}{q\left(\theta\right)}c.\label{eq:app:w-01}
\end{equation}

The key Bellman equations for workers are 
\begin{equation}
\mathcal{E}=w+\beta\left[sU+\left(1-s\right)\mathcal{E}\right]\label{eq:app:E}
\end{equation}
\begin{equation}
\mathcal{U}=z+\beta\left\{ \theta q\left(\theta\right)\mathcal{E}+\left[1-\theta q\left(\theta\right)\right]\mathcal{U}\right\} .\label{eq:app:U}
\end{equation}

In the canonical matching model, the match surplus,
\begin{align*}
\mathcal{S} = \left(\mathcal{J}-\mathcal{V}\right)+\left(\mathcal{E}-\mathcal{U}\right)
\end{align*}
is
the benefit a firm gains from a productive match over an unfilled vacancy
plus
the benefit a worker gains from employment over unemployment.
The surplus is split between a matched firm--worker pair.
The outcome of Nash bargaining specifies 
\begin{equation}
\mathcal{E}-\mathcal{U}=\phi\mathcal{S}\text{ and }\mathcal{J}=\left(1-\phi\right)\mathcal{S},\label{eq:app:Nash-outcome}
\end{equation}
where $\phi\in\left[0,1\right)$ measures the worker's bargaining power.

\subsection{On the Value of Unemployment}
\label{sec:app:value-unemployment}

The next part of the derivation yields a value for unemployment.
Solving equation \eqref{eq:app:J} for $\mathcal{J}$ yields 
\begin{align*}
\mathcal{J} & =y-w+\beta\left(1-s\right)\mathcal{J}\\
\therefore\mathcal{J}\left[1-\beta\left(1-s\right)\right] & =y-w\\
\therefore\mathcal{J} & =\frac{y-w}{1-\beta\left(1-s\right)}.
\end{align*}
And solving \eqref{eq:app:E} for $\mathcal{E}$ yields 
\begin{align*}
\mathcal{E} &= w+\beta s\mathcal{U}+\beta\left(1-s\right)\mathcal{E}\\
\therefore\mathcal{E}\left[1-\beta\left(1-s\right)\right] & =w+\beta s\mathcal{U}\\
\therefore\mathcal{E} & =\frac{w+\beta s\mathcal{U}}{1-\beta\left(1-s\right)}\\
 & =\frac{w}{1-\beta\left(1-s\right)}+\frac{\beta s\mathcal{U}}{1-\beta\left(1-s\right)}.
\end{align*}
Developing the expressions in \eqref{eq:app:Nash-outcome} for the outcome of Nash bargaining yields 
\begin{align*}
\mathcal{E}-\mathcal{U} & =\phi\mathcal{S}\\
 & =\phi\frac{\mathcal{J}}{1-\phi}
\end{align*}
and using the just-derived expressions for $\mathcal{J}$ and $\mathcal{E}$
yields
\begin{align*}
\underbrace{\left[\frac{w}{1-\beta\left(1-s\right)}+\frac{\beta s\mathcal{U}}{1-\beta\left(1-s\right)}\right]}_{\mathcal{E}}-\mathcal{U}=\frac{\phi}{1-\phi}\underbrace{\left[\frac{y-w}{1-\beta\left(1-s\right)}\right]}_{\mathcal{J}}.
\end{align*}
Developing this expression yields
\begin{align}
  \label{eq:app:w-02}
  \begin{split}
w+\beta s\mathcal{U}-\left[1-\beta\left(1-s\right)\right]\mathcal{U} &= \frac{\phi}{1-\phi}\left(y-w\right) \\
\therefore w+\beta s\mathcal{U}-\mathcal{U}+\beta\mathcal{U}-s\beta\mathcal{U} & =\frac{\phi}{1-\phi}\left(y-w\right) \\
\therefore w & =\frac{\phi}{1-\phi}\left(y-w\right)+\left(1-\beta\right)\mathcal{U} \\
\therefore\left(1-\phi\right)w & =\phi\left(y-w\right)+\left(1-\phi\right)\left(1-\beta\right)\mathcal{U} \\
\therefore w & =\phi y+\left(1-\beta\right)\mathcal{U}-\phi\left(1-\beta\right)\mathcal{U}.  
  \end{split}
\end{align}
Using the fact that
\begin{align*}
1-\beta=1-\frac{1}{1+r}=\frac{1+r-1}{1+r}=\frac{r}{1+r},
\end{align*}
the latter expression can be written as 
\begin{equation}
w=\frac{r}{1+r}\mathcal{U}+\phi\left(y-\frac{r}{1+r}\mathcal{U}\right),\label{eq:w-02}
\end{equation}
which is equation (9) in \citet[2634]{ljungqvist_sargent_2017}.
The value $r\mathcal{U}/\left(1+r\right)$ in equation \eqref{eq:w-02} is the ``annuity value of being unemployed'' \citep[][2634]{ljungqvist_sargent_2017}.

To get an expression for the annuity value of unemployment,
$r\mathcal{U}/\left(1+r\right)$,
I solve equation \eqref{eq:app:U} for $\mathcal{E}-\mathcal{U}$ and
substitute this expression and the expression in \eqref{eq:app:J-eqm} into \eqref{eq:app:Nash-outcome}.

These steps are taken next:

Turning to equation \eqref{eq:app:U}: 
\begin{align*}
\mathcal{U} & =z+\beta\left\{ \theta q\left(\theta\right)\mathcal{E}+\left[1-\theta q\left(\theta\right)\right]\mathcal{U}\right\} \\
\therefore\mathcal{U} & =z+\beta\theta q\left(\theta\right)\mathcal{E}+\beta\mathcal{U}-\beta\theta q\left(\theta\right)\mathcal{U}\\
\therefore\mathcal{U} & =z+\left[\beta\theta q\left(\theta\right)\right]\left(\mathcal{E}-\mathcal{U}\right)+\beta\mathcal{U}\\
\mathcal{U}\left(1-\beta\right)-z & =\left[\beta\theta q\left(\theta\right)\right]\left(\mathcal{E}-\mathcal{U}\right)\\
\therefore\mathcal{E}-\mathcal{U} & =\frac{1}{\beta\theta q\left(\theta\right)}\left[\left(1-\beta\right)\mathcal{U}-z\right]\\
 & =\frac{1+r}{\theta q\left(\theta\right)}\left[\left(1-\beta\right)\mathcal{U}-z\right]\\
 & =\frac{r}{\theta q\left(\theta\right)}\mathcal{U}-\frac{1+r}{\theta q\left(\theta\right)}z.
\end{align*}
Using this expression for $\mathcal{E}-\mathcal{U}$ in \eqref{eq:app:Nash-outcome} yields 
\begin{align*}
\mathcal{E}-\mathcal{U} & =\phi\mathcal{S}\\
\therefore\frac{r}{\theta q\left(\theta\right)}\mathcal{U}-\frac{1+r}{\theta q\left(\theta\right)}z & =\phi\mathcal{S}\\
 & =\phi\left(\frac{\mathcal{J}}{1-\phi}\right)\\
 & =\frac{\phi}{1-\phi}\frac{c}{\beta q\left(\theta\right)},
\end{align*}
where the last equality uses the expression for $\mathcal{J}$ in
equation \eqref{eq:app:J-eqm}.
Developing this expression yields
\begin{align}
\frac{r}{\theta q\left(\theta\right)}\mathcal{U}-\frac{1+r}{\theta q\left(\theta\right)}z & =\frac{\phi}{1-\phi}\frac{c}{\beta q\left(\theta\right)}\nonumber \\
\therefore r\mathcal{U}-\left(1+r\right)z & =\frac{\phi}{1-\phi}\frac{1}{\beta}c\theta\nonumber \\
\therefore r\mathcal{U}-\left(1+r\right)z & =\frac{\phi}{1-\phi}\left(1+r\right)c\theta\nonumber \\
\therefore\frac{r}{1+r}\mathcal{U} & =z+\frac{\phi c\theta}{1-\phi},\label{eq:app:U-annuity}
\end{align}
which is equation (10) in \citet[2634]{ljungqvist_sargent_2017}.

Substituting equation \eqref{eq:app:U-annuity} into equation \eqref{eq:w-02} yields an expression for the wage: 
\begin{align*}
w & =\frac{r}{1+r}\mathcal{U}+\phi\left(y-\frac{r}{1+r}\mathcal{U}\right)\\
 & =z+\frac{\phi c\theta}{1-\phi}+\phi\left(y-z-\frac{\phi c\theta}{1-\phi}\right)\\
 & =\left(1-\phi\right)z+\left(1-\phi\right)\left(\frac{\phi c\theta}{1-\phi}\right)+\phi y\\
 & =\left(1-\phi\right)z+\phi c\theta+\phi y
\end{align*}
or 
\begin{equation}
w=z+\phi\left(y-z+\theta c\right),\label{eq:app:w-03}
\end{equation}
which is equation (11) in \citet[2635]{ljungqvist_sargent_2017}.

\subsection{Existence and Uniqueness}
\label{sec:app:existence-uniqueness}

The two expressions for the wage rate in \eqref{eq:app:w-01} and \eqref{eq:app:w-03} jointly determine the equilibrium value of $\theta$:
\begin{align*}
y-\frac{r+s}{q\left(\theta\right)}c & =z+\phi\left(y-z+\theta c\right).
\end{align*}
Developing this expression yields 
\begin{align*}
y-\frac{r+s}{q\left(\theta\right)}c & =z+\phi\left(y-z+\theta c\right)\\
\therefore y-z & =\phi\left(y-z+\theta c\right)+\frac{r+s}{q\left(\theta\right)}c\\
\therefore\left(1-\phi\right)\left(y-z\right) & =\phi\theta c+\frac{r+s}{q\left(\theta\right)}c\\
 & =\left[\frac{\phi\theta q\left(\theta\right)}{q\left(\theta\right)}+\frac{r+s}{q\left(\theta\right)}\right]c,
\end{align*}
which can be re-arranged to yield an expression in $\theta$ alone:
\begin{equation}
y-z=\frac{r+s+\phi\theta q\left(\theta\right)}{\left(1-\phi\right)q\left(\theta\right)}c.\label{eq:app:eqm-theta}
\end{equation}
Equation \eqref{eq:app:eqm-theta} implicitly defines an equilibrium
level of tightness.
The expression agrees with equation (12) in
\citet[][2635]{ljungqvist_sargent_2017}.
\citet{pissarides_2000} also shows how similar equations can be manipulated to yield a single expression in $\theta$ alone.

Existence and uniqueness of equilibrium tightness is established by
proposition \ref{prop:app:unique-theta}. 
\begin{prop}
\label{prop:app:unique-theta} Suppose $y>z$, which says that workers
produce more of the homogeneous consumption good at work than at home,
and suppose that $\left(1-\phi\right)\left(y-z\right)/\left(r+s\right)>c$.
Then a unique $\theta>0$ solves the relationship in \eqref{eq:app:eqm-theta}.

The condition that $\left(1-\phi\right)\left(y-z\right)/\left(r+s\right)>c$
requires that the value of the first job opening is positive.

The steady-state level of unemployment is 
\begin{equation}
u=\frac{s}{s+f\left(\theta\right)}=\frac{s}{s+\theta q\left(\theta\right)}.\label{eq:steady-state-u}
\end{equation}
\end{prop}

\begin{proof}
I establish proposition \ref{prop:app:unique-theta} in three steps:
\end{proof}
\begin{enumerate}
\item \label{item:unique-theta:existence-uniqueness} Existence and uniqueness of the economy's steady-state equilibrium
are established.
\item \label{item:unique-theta:initial-vacancy} The required condition on parameter values that guarantees an equilibrium
is then interpreted as the positive value of posting an initial vacancy,
offering an alternative interpretation from the one given by \citet{pissarides_2000}.
\item \label{item:unique-theta:eqm-ur} When the number of jobs created equals the number of jobs destroyed,
the familiar expression for steady-state unemployment depends on the
rate of separation to the sum of the rates of separation and finding.
This result is well known and is included for completeness \citep[see, for example, ][]{pissarides_2000}. 
\end{enumerate}

\begin{proof}
  \ul{\mbox{Step \ref{item:unique-theta:existence-uniqueness}}}:
  To establish existence of an equilibrium, I define the function 
\begin{align*}
\mathcal{T}\left(\tilde{\theta}\right) & =\frac{y-z}{c}-\frac{r+s+\phi\tilde{\theta}q\left(\tilde{\theta}\right)}{\left(1-\phi\right)q\left(\tilde{\theta}\right)}\\
 & =\frac{y-z}{c}-\frac{r+s}{\left(1-\phi\right)q\left(\tilde{\theta}\right)}-\frac{\phi}{1-\phi}\tilde{\theta}.
\end{align*}
Then, using the fact that $\lim_{\tilde{\theta}\rightarrow0}q\left(\tilde{\theta}\right)=1$,
\[
\lim_{\tilde{\theta}\rightarrow0}\mathcal{T}\left(\tilde{\theta}\right)=\frac{y-z}{c}-\frac{r+s}{1-\phi}>0,
\]
where the inequality uses the assumption that $\left(1-\phi\right)\left(y-z\right)/\left(r+s\right)>0$.
Additionally, I define 
\[
\tilde{\theta}^{\bullet}=\frac{1-\phi}{\phi}\frac{y-z}{c}>0,
\]
which is positive because $y>z$ and $\phi\in\left[0,1\right)$. Then
\begin{align*}
\mathcal{T}\left(\tilde{\theta}^{\bullet}\right) & =\frac{y-z}{c}-\frac{r+s}{\left(1-\phi\right)q\left(\tilde{\theta}^{\bullet}\right)}-\frac{\phi}{1-\phi}\tilde{\theta}^{\bullet}\\
 & =\frac{y-z}{c}-\frac{r+s}{\left(1-\phi\right)q\left(\tilde{\theta}^{\bullet}\right)}-\frac{\phi}{1-\phi}\frac{1-\phi}{\phi}\frac{y-z}{c}\\
 & =-\frac{r+s}{\left(1-\phi\right)q\left(\tilde{\theta}^{\bullet}\right)}\\
 & <0,
\end{align*}
where the inequality follows from the fact that $q\left(\tilde{\theta}^{\bullet}\right)>0$.
Because $\mathcal{T}$ is a combination of continuous functions, it
is also continuous. Therefore, an application of the intermediate
value theorem establishes that there exists $\theta^{\star}\in\left(0,\frac{1-\phi}{\phi}\frac{y-z}{c}\right)$
such that $\mathcal{T}\left(\theta^{\star}\right)=0$.

The uniqueness part of proposition \ref{prop:app:unique-theta} comes
from the fact that $\mathcal{T}$ is decreasing. Indeed, 
\[
\mathcal{T}^{\prime}\left(\tilde{\theta}\right)=\frac{r+s}{\left(1-\phi\right)\left[q\left(\tilde{\theta}\right)\right]^{2}}q^{\prime}\left(\tilde{\theta}\right)-\frac{\phi}{\left(1-\phi\right)}<0,
\]
and the inequality comes from the fact that $q^{\prime}<0$.

\ul{\mbox{Step \ref{item:unique-theta:initial-vacancy}}}:
The condition that $\left(1-\phi\right)\left(y-z\right)/\left(r+s\right)>c$ restates
the requirement that the initial vacancy is profitable.
The following thought experiment demonstrates why.

Starting from a given level of unemployment, which is guaranteed with
exogenous separations, the value of posting an initial vacancy
is computed as $\lim_{\theta\rightarrow0}\mathcal{V}.$ In this thought
experiment, the probability that the vacancy is filled is $1$, as
$\lim_{\theta\rightarrow0}q\left(\theta\right)=1$. The following
period the firm earns the value of a productive match, which equals
the flow payoff $y-w$ plus the value of a productive match discounted
by $\beta\left(1-s\right)$: $\mathcal{J}=y-w+\beta\left(1-s\right)\mathcal{J}$.
Solving this expression for $\mathcal{J}$ yields 
\begin{align*}
\mathcal{J} & =y-w+\beta\left(1-s\right)\mathcal{J}\\
\therefore\mathcal{J}\left[1-\beta\left(1-s\right)\right] & =y-w\\
\therefore\mathcal{J} & =\frac{y-w}{1-\beta\left(1-s\right)}.
\end{align*}

The wage rate paid by the firm, looking at the expression in \eqref{eq:app:w-03}, is 
\begin{align*}
\lim_{\theta\rightarrow0}w & =\lim_{\theta\rightarrow0}z+\phi\left(y-z+\theta c\right)\\
 & =\phi y+\left(1-\phi\right)z,
\end{align*}
making
\begin{align*}
\mathcal{J}=\frac{\left(1-\phi\right)\left(y-z\right)}{1-\beta\left(1-s\right)}.
\end{align*}
Using this expression in the value of an initial vacancy 
\begin{align*}
\lim_{\theta\rightarrow0}\mathcal{V} & =\lim_{\theta\rightarrow0}\left\langle -c+\beta\left\{ q\left(\theta\right)\mathcal{J}+\left[1-q\left(\theta\right)\right]\mathcal{V}\right\} \right\rangle \\
 & =-c+\beta\frac{\left(1-\phi\right)\left(y-z\right)}{1-\beta\left(1-s\right)}\\
 & >0.
\end{align*}
The inequality stipulates that in order to start the process of posting
vacancies, the first vacancy needs to be profitable. Developing this inequality yields 
\begin{align*}
\beta\frac{\left(1-\phi\right)\left(y-z\right)}{1-\beta\left(1-s\right)} & >c\\
\therefore\frac{1}{1+r}\frac{\left(1-\phi\right)\left(y-z\right)}{1-\frac{1}{1+r}\left(1-s\right)} & >c\\
\therefore\frac{\left(1-\phi\right)\left(y-z\right)}{1+r-1+s} & >c
\end{align*}
or
\begin{align*}
\frac{\left(1-\phi\right)\left(y-z\right)}{r+s}>c.
\end{align*}
Which is the condition listed in proposition \ref{prop:app:unique-theta}.

\ul{\mbox{Step \ref{item:unique-theta:eqm-ur}}}: The
steady-state level of unemployment comes from the evolution of unemployment and steady-state equilibrium where $u_{t+1}=u_{t}=u$ and $\theta_{t} = \theta$.
Next period's unemployment comprises
unemployed workers who did not find a job, $\left[1-f\left(\theta_{t}\right)\right]u_{t}$, plus
employed workers who separate from jobs, $s\left(1-u_{t}\right)$, where $1-u_{t}$ is the level of employment after
the normalization that the size of the labor force equal one.
From this law of motion:
\begin{align*}
u_{t+1} & =\left[1-f\left(\theta_{t}\right)\right]u_{t}+s\left(1-u_{t}\right)\\
\therefore u & =\left[1-f\left(\theta\right)\right]u+s\left(1-u\right)\\
\therefore0 & =-f\left(\theta\right)u+s-su\\
\therefore\left(s+f\right)u & =s
\end{align*}
or $u=s/\left[s+f\left(\theta\right)\right]$, establishing \eqref{eq:steady-state-u}. 
\end{proof}

\subsection{Joint Parameterization of $c$ and $A$}
\label{sec:joint-param-c-A}

The joint parameterization of the cost of posting a vacancy, $c$, and matching efficiency, $A$, is a
``choice of normalization'' \citep[p~12, footnote 33 of the accompanying online appendix to][]{ljungqvist_sargent_2017}.
Given a calibration $c$ and $A$,
which produces an equilibrium market tightness through the implicit expression for $\theta$ in \eqref{eq:app:eqm-theta},
the same equilibrium market tightness and job-finding probability can be attained with an alternative parameterization, $\hat{c}$ and $\hat{A}$.

See \citet{kiarsi_2020} for the importance of the cost of posting a vacancy.

To verify this claim,
I differentiate matching technologies by explicitly referencing the matching efficiency parameter,
which affects the technologies as a multiplicative constant.
Both \eqref{eq:app:M-cobb-douglas} and \eqref{eq:app:M-nonlinear} can be expressed this way.
The two job-filling rates, for example, will be $Aq \left( \theta \right)$        and $\hat{A}        q \left( \theta \right)$.
The two job-finding rates, for example, will be $A\theta q \left( \theta \right)$ and $\hat{A} \theta q \left( \theta \right)$.

I start by letting $\hat{c} = \zeta c$ for some $\zeta > 0$ and $\zeta < \left[ A q_{A} \left( \theta \right) \right]^{-1}$.
The expression for equilibrium market tightness in \eqref{eq:app:eqm-theta} becomes 
\begin{align*}
y-z & =\frac{r+s+\phi\hat{A}\hat{\theta}q \left(\hat{\theta}\right)}{\left(1-\phi\right)\hat{A}q \left(\hat{\theta}\right)}\hat{c}\\
\therefore\frac{1-\phi}{\zeta c}\left(y-z\right) & =\frac{r+s}{\hat{A}q \left(\hat{\theta}\right)}+\phi\hat{\theta}\\
\therefore\frac{\left(1-\phi\right)\left(y-z\right)}{c} & =\zeta\frac{r+s}{\hat{A}q\left(\hat{\theta}\right)}+\phi\zeta\hat{\theta}.
\end{align*}
Comparing this to the original parameterization yields
\begin{align*}
\frac{\left(1-\phi\right)\left(y-z\right)}{c} & =\zeta\frac{r+s}{\hat{A}q\left(\hat{\theta}\right)}+\phi\zeta\hat{\theta}=\frac{r+s}{Aq \left(\theta\right)}+\phi\theta.
\end{align*}
For these to be equal, it must be the case that
\begin{enumerate}
\item\label{item:1} $\zeta \hat{\theta} = \theta$ and
\item\label{item:2} $\hat{A}q \left(\hat{\theta}\right)\frac{1}{\zeta} = Aq \left(\theta\right)$.
\end{enumerate}
Condition \ref{item:1} implies $\hat{\theta} = \theta / \zeta$.
Condition \ref{item:2} implies
\begin{align*}
  \hat{A}q \left(\hat{\theta}\right)\frac{1}{\zeta} &= Aq \left(\theta\right) \\
\therefore\hat{A}q \left(\theta/\zeta\right) & =\zeta Aq \left(\theta\right)\\
\therefore\hat{A} & =\zeta\frac{Aq \left(\theta\right)}{q \left(\theta/\zeta\right)}.
\end{align*}
The job-finding rate is the same:
\begin{align*}
\hat{\theta}\hat{A}q \left(\hat{\theta}\right) &= \left(\frac{\theta}{\zeta}\right)\zeta\frac{Aq \left(\theta\right)}{q\left(\theta/\zeta\right)}q\left(\frac{\theta}{\zeta}\right)\\
 &= A\theta q \left(\theta\right).
\end{align*}
The value for job creation, from \eqref{eq:app:J-eqm}, is also the same:
\begin{align*}
\hat{J} & =\frac{\hat{c}}{\beta\hat{A}q_{A}\left(\hat{\theta}\right)}\\
 & =\frac{\zeta c}{\beta\zeta\frac{Aq_{A}\left(\theta\right)}{q_{A}\left(\theta/\zeta\right)}q_{A}\left(\theta/\zeta\right)}\\
 & =\frac{c}{\beta Aq_{A}\left(\theta\right)} \\
  &= J.
\end{align*}

The job-filling probability is proportional to the original job-filling probability:
\begin{align*}
\hat{A}q_{A}\left(\hat{\theta}\right)=\zeta Aq_{A}\left(\theta\right).
\end{align*}

The choice of $\zeta$ must also be careful to not push $\zeta Aq \left(\theta\right)$ outside of $\left(0,1\right)$,
which is guaranteed if
\begin{align*}
0 < \zeta < \frac{1}{A q_{A} \left( \theta \right)}.
\end{align*}
When the matching technology takes the Cobb--Douglas form,
this condition amounts to $\zeta < \theta^{\alpha} / A$,
which is the condition reported by \citet{ljungqvist_sargent_2017} in their online appendix.

When the matching technology is Cobb--Douglas, then the $\hat{\theta}=\theta/\zeta$ and
\begin{align*}
\hat{A} & =\zeta\frac{Aq_{A}\left(\theta\right)}{q_{A}\left(\theta/\zeta\right)}\\
 & =\zeta\frac{A\theta^{-\alpha}}{\left(\frac{\theta}{\zeta}\right)^{-\alpha}}\\
 & =\zeta\frac{A\theta^{-\alpha}}{\theta^{-\alpha}\zeta^{\alpha}}\\
 & =\zeta A\zeta^{-\alpha}\\
 & =\zeta^{1-\alpha}A,
\end{align*}
which is reported in footnote 33 of the online appendix to \citet{ljungqvist_sargent_2017}.

When the matching technology takes the form of the nonlinear case given in \eqref{eq:app:M-nonlinear}, a case not considered by \citet{ljungqvist_sargent_2017},
where $q \left(\theta\right) = A\left(1+\theta^{\gamma}\right)^{-1/\gamma}$,
then $\hat{\theta}=\theta/\zeta$ and
\begin{align*}
\hat{A} & =\zeta\frac{Aq_{A}\left(\theta\right)}{q_{A}\left(\theta/\zeta\right)}\\
 & =\zeta\frac{A\left(1+\theta^{\gamma}\right)^{-1/\gamma}}{\left[1+\left(\theta/\zeta\right)^{\gamma}\right]^{-1/\gamma}}.
\end{align*}

\subsection{A Decomposition of the Elasticity of Market Tightness and\\the Fundamental Surplus}
\label{sec:decomp-elast-mark}

This section follows section II.A of \citet{ljungqvist_sargent_2017}, beginning on page 2636.

The elasticity of tightness with respect to productivity is
\begin{align*}
\eta_{\theta,y} \coloneq \frac{d\theta}{dy}\frac{\theta}{y}\approx\frac{\Delta\theta/\theta}{\Delta y/y},
\end{align*}
where the approximation indicates that we are talking about the percent
change in tightness with respect to the percent change in $y$. Following
\citet{ljungqvist_sargent_2017}, in this section, I decompose this
key elasticity into two factors:
The first is a factor bounded away from $1$ and the inverse of the elasticity of matching with respect to unemployment.
The second is the inverse of fundamental surplus fraction.

To uncover $\eta_{\theta,y}$, note that equation \eqref{eq:app:eqm-theta}
can be written 
\begin{align}
  \label{eq:app:implicit-theta}
  \begin{split}
y-z & =\frac{r+s+\phi\theta q\left(\theta\right)}{\left(1-\phi\right)q\left(\theta\right)}c \\
\therefore\frac{1-\phi}{c}\left(y-z\right) & =\frac{r+s+\phi\theta q\left(\theta\right)}{q\left(\theta\right)} \\
 & =\frac{r+s}{q\left(\theta\right)}+\phi\theta.    
  \end{split}
\end{align}
Define
\begin{align*}
\digamma\left(\theta,y\right) \coloneq \frac{1-\phi}{c}\left(y-z\right)-\frac{r+s}{q\left(\theta\right)}-\phi\theta.
\end{align*}
Then implicit differentiation implies 
\begin{align*}
\frac{d\theta}{dy} & =-\frac{\partial\digamma/\partial y}{\partial\digamma/\partial\theta}=-\frac{\frac{1-\phi}{c}}{\frac{r+s}{\left[q\left(\theta\right)\right]^{2}}q^{\prime}\left(\theta\right)-\phi}\\
 & =-\frac{\left[\frac{r+s}{q\left(\theta\right)}+\phi\theta\right]\frac{1}{y-z}}{\frac{r+s}{\left[q\left(\theta\right)\right]^{2}}q^{\prime}\left(\theta\right)-\phi},
\end{align*}
where the last equality uses the equality in \eqref{eq:app:eqm-theta}.
Developing this expression yields 
\begin{align}
  \label{eq:app:devel-01}    
  \begin{split}
\frac{d\theta}{dy} & =-\frac{\left[\frac{r+s}{q\left(\theta\right)}+\phi\theta\right]}{\frac{r+s}{\left[q\left(\theta\right)\right]^{2}}q^{\prime}\left(\theta\right)-\phi}\frac{1}{y-z}\times\frac{\theta q\left(\theta\right)}{\theta q\left(\theta\right)} \\
 & =-\frac{\left[r+s+\phi\theta q\left(\theta\right)\right]}{\left(r+s\right)\frac{q^{\prime}\left(\theta\right)\theta}{q\left(\theta\right)}-\phi\theta q\left(\theta\right)}\frac{\theta}{y-z}.
  \end{split}
\end{align}
The expression in the denominator is related to the elasticity of matching with respect to unemployment.

Using the expression for $\eta_{M,u}$ in \eqref{eq:app:eta-M-u} in the
developing expression for $d\theta/dy$ yields 
\begin{align}
  \label{eq:app:dtheta-dy}
  \begin{split}
\frac{d\theta}{dy} & =-\frac{\left[r+s+\phi\theta q\left(\theta\right)\right]}{\left(r+s\right)\left(\frac{q^{\prime}\left(\theta\right)\theta}{q\left(\theta\right)}\right)-\phi\theta q\left(\theta\right)}\frac{\theta}{y-z} \\
 & =\frac{r+s+\phi\theta q\left(\theta\right)}{\left(r+s\right)\eta_{M,u}+\phi\theta q\left(\theta\right)}\frac{\theta}{y-z}.    
  \end{split}
\end{align}
Further developing this expression yields 
\begin{align*}
\frac{d\theta}{dy} & =\frac{r+s+\phi\theta q\left(\theta\right)}{\left(r+s\right)\eta_{M,u}+\phi\theta q\left(\theta\right)}\frac{\theta}{y-z}\\
 & =\frac{r+s+\left(r+s\right)\eta_{M,u}-\left(r+s\right)\eta_{M,u}+\phi\theta q\left(\theta\right)}{\left(r+s\right)\eta_{M,u}+\phi\theta q\left(\theta\right)}\frac{\theta}{y-z}\\
 & =\left[1+\frac{r+s-\left(r+s\right)\eta_{M,u}}{\left(r+s\right)\eta_{M,u}+\phi\theta q\left(\theta\right)}\right]\frac{\theta}{y-z}\\
 & =\left[1+\frac{\left(r+s\right)\left(1-\eta_{M,u}\right)}{\left(r+s\right)\eta_{M,u}+\phi\theta q\left(\theta\right)}\right]\frac{\theta}{y-z}.
\end{align*}

And this expression implies 
\begin{align}
  \label{eq:app:decomp}
  \begin{split}
    \eta_{\theta,y} & =\frac{d\theta}{dy}\frac{y}{\theta} \\
 & =\left[1+\frac{\left(r+s\right)\left(1-\eta_{M,u}\right)}{\left(r+s\right)\eta_{M,u}+\phi\theta q\left(\theta\right)}\right]\frac{y}{y-z} \\
 & =\Upsilon\frac{y}{y-z},
  \end{split}
\end{align}
which is a fundamental result in \citet{ljungqvist_sargent_2017},
expressed in their equation (15) on page 2636.
The expression in \eqref{eq:app:decomp} decomposes the elasticity of tightness with respect to productivity
into two factors. \citet{ljungqvist_sargent_2017} focus on the second
factor because the first factor, $\Upsilon$, ``has an upper bound
coming from a consensus about values of the elasticity of matching
with respect to unemployment.''

As long as the conditions for an interior equilibrium in proposition
\ref{prop:unique-theta} hold, $\Upsilon$ is bounded above by $1/\eta_{M,u}$.
\citet{ljungqvist_sargent_2017} establish this fact by noting that
$\Upsilon$ in \eqref{eq:app:decomp} can be written as
\begin{align*}
\Upsilon\left(\chi\right)=\left[1+\frac{\left(r+s\right)\left(1-\eta_{M,u}\right)}{\left(r+s\right)\eta_{M,u}+\chi\phi\theta q\left(\theta\right)}\right]
\end{align*}
and noting that $\Upsilon\left(\chi\right)$ can be viewed as a function
of $\chi$ and equal to $\Upsilon$ when $\chi=1$.
Evaluating $\Upsilon\left(\chi\right)$ at $\chi=0$ implies 
\begin{align*}
\Upsilon\Big\vert_{\chi=0} & =1+\frac{\left(r+s\right)\left(1-\eta_{M,u}\right)}{\left(r+s\right)\eta_{M,u}}=1+\frac{1-\eta_{M,u}}{\eta_{M,u}}\\
 & =\frac{1}{\eta_{M,u}}>1,
\end{align*}
where the inequality is established in proposition \ref{prop:eta-M-u}.
Moreover, $\Upsilon$ is decreasing in $\chi$:
\begin{align*}
\frac{\partial\Upsilon}{\partial\chi}=-\frac{\left(r+s\right)\left(1-\eta_{M,u}\right)}{\left[\left(r+s\right)\eta_{M,u}+\chi\phi\theta q\left(\theta\right)\right]^{2}}\chi\theta q\left(\theta\right)<0.
\end{align*}
Thus
\begin{align*}
1<\Upsilon \coloneq \Upsilon\left(1\right)<\Upsilon\left(0\right)=\frac{1}{\eta_{M,u}}.
\end{align*}

These two facts establish that $\Upsilon$ is bounded above by $1/\eta_{M,u}$.
Moreover, the expression for $\Upsilon$, defined in \eqref{eq:app:decomp},
establishes that $\Upsilon$ is bounded below by $1$.
The results are collected in proposition \ref{prop:app:upsilon}.
\begin{prop}
  \label{prop:app:upsilon}
  In the canonical DMP search model,
  which features
  a general matching technology,
  random search, linear utility, workers with identical capacities for work,
  exogenous separations, and no disturbances in aggregate productivity,
  the elasticity of market tightness with respect to productivity can be decomposed as  
\begin{align*}
\eta_{\theta,y}=\Upsilon\frac{y}{y-z},
\end{align*}
where the second factor is the inverse of fundamental surplus fraction
and the first factor is bounded below by $1$ and above by $1/\eta_{M,u}$:
\begin{align*}
1<\Upsilon<\frac{1}{\eta_{M,u}}.
\end{align*}
\end{prop}

Proposition \ref{prop:eta-M-u} establishes that $1 / \eta_{M,u}$ is larger than one.

Proposition \ref{prop:app:upsilon} is the same as proposition \ref{prop:upsilon},
which is stated in the text.

\subsection{Wage Elasticity in the Canonical Model}

The elasticity of wages with respect to productivity is $dw/dy\times y/w=:\eta_{w,y}$:
\begin{align*}
\frac{\frac{\Delta w}{w}}{\frac{\Delta y}{dy}}\approx\frac{dw}{dy}\frac{y}{w}=:\eta_{w,y}.
\end{align*}
An expression for $dw/dy$ can be derived from \eqref{eq:app:w-03}:
\begin{align*}
w & =z+\phi\left(y-z+\theta c\right)\\
\therefore\frac{dw}{dy} & =\phi\left(1+c\frac{d\theta}{dy}\right)\\
 & =\phi\left[1+c\left(\frac{r+s+\phi\theta q\left(\theta\right)}{\left(r+s\right)\eta_{M,u}+\phi\theta q\left(\theta\right)}\right)\frac{\theta}{y-z}\right],
\end{align*}
where the last inequality uses the expression for $d\theta/dy$ in
\eqref{eq:app:dtheta-dy}. The equilibrium condition in \eqref{eq:app:eqm-theta}
implies
\[
\frac{y-z}{c}=\frac{r+s+\phi\theta q\left(\theta\right)}{\left(1-\phi\right)q\left(\theta\right)}.
\]
Using this expression in the developing expression for $dw/dy$ yields
\begin{align*}
\frac{dw}{dy} & =\phi\left[1+\left(\frac{r+s+\phi\theta q\left(\theta\right)}{\left(r+s\right)\eta_{M,u}+\phi\theta q\left(\theta\right)}\right)\frac{\left(1-\phi\right)\theta q\left(\theta\right)}{r+s+\phi\theta q\left(\theta\right)}\right]\\
 & =\phi\left[1+\frac{\left(1-\phi\right)\theta q\left(\theta\right)}{\left(r+s\right)\eta_{M,u}+\phi\theta q\left(\theta\right)}\right]\\
 & =\phi\left[\frac{\left(r+s\right)\eta_{M,u}+\phi\theta q\left(\theta\right)+\left(1-\phi\right)\theta q\left(\theta\right)}{\left(r+s\right)\eta_{M,u}+\phi\theta q\left(\theta\right)}\right]
\end{align*}
or
\begin{equation}
\frac{dw}{dy}=\phi\left[\frac{\left(r+s\right)\eta_{M,u}+\theta q\left(\theta\right)}{\left(r+s\right)\eta_{M,u}+\phi\theta q\left(\theta\right)}\right].\label{eq:app:dw-dy}
\end{equation}
When the worker has no bargaining power, $\phi=0$, then $dw/dy$
evaluates to $0$. When the worker has all the bargaining power, $\phi=1$,
then $dw/dy=1$.

\clearpage
\pagebreak

\section{Converting a Daily Job-finding Rate to a Monthly Job-finding Rate}
\label{sec:app:converting-daily-to-monthly}

Here I go through calculations that convert a daily rate to a monthly rate.

The probably a worker finds a job within the month is $1$ minus the probably they do not find a job.
What is the probability that a worker does not find a job?
For the first day of the month,
the probability of not finding a job is $1-\theta q\left(\theta\right)$, where $\theta q \left( \theta \right)$ is the daily job-finding probability.
Over two days, the probability is not finding a job on day $1$ and on day $2$,
which is $\left[1-\theta q\left(\theta\right)\right]^{2}$.
Over the whole month, the probability of not finding a job is $\left[1-\theta q\left(\theta\right)\right]^{30}$.
Thus, the monthly job-finding probability is
\begin{align*}
\text{monthly job-finding rate}=1-\left[1-\theta q\left(\theta\right)\right]^{30}.
\end{align*}
A similar calculation converts the daily job-filling rate to the monthly job-filling rate:
\begin{align*}
\text{monthly job-filling rate}=1-\left[1-q\left(\theta\right)\right]^{30}.
\end{align*}

\end{document}